\documentclass[english]{article}
\usepackage{authblk}

\usepackage[utf8]{inputenc}
\usepackage{microtype}

\usepackage[T1]{fontenc}
\usepackage{libertine}
\usepackage[table]{xcolor}

\usepackage{amssymb,amsthm}
\usepackage{bm}

\newtheorem{thm}{Theorem}[section]
\newtheorem*{thm*}{Theorem}
\newtheorem{lme}{Lemma}[section]
\newtheorem{prop}{Proposition}[section]

\newtheorem{cor}{Corollary}[section]

\usepackage[format=plain,font=it]{caption}
\usepackage{algorithm}
\usepackage{algpseudocode}
\makeatletter
\def\plist@algorithm{Alg.\space} 
\makeatother
\usepackage{pdflscape}
\usepackage{enumitem}
\usepackage{float}
\usepackage{subcaption}
\usepackage{amsmath}
\usepackage{amsfonts}
\usepackage{comment}
\usepackage{dsfont}
\newtheorem*{remark}{Remark}

\usepackage{mathtools} 
\usepackage{booktabs} 
\usepackage{tikz} 
\usetikzlibrary{arrows.meta, positioning, calc}
\usepackage{multirow}

\usepackage{array}
\newcolumntype{C}[1]{>{\centering\arraybackslash}m{#1}}
\newcolumntype{R}[1]{>{\raggedleft\arraybackslash}m{#1}}
\newcolumntype{L}[1]{>{\raggedright\arraybackslash}m{#1}}

\usepackage{graphicx,pstricks,psfrag}
\definecolor{grisclair}{HTML}{A7A7A7}
\definecolor{bleu}{RGB}{0,101,189}
\definecolor{vert}{HTML}{004D40}
\definecolor{rose}{HTML}{D81B60}
\definecolor{bleuTOL}{HTML}{332288}
\definecolor{wongBlack}{RGB}{0,0,0}
\definecolor{wongGold}{RGB}{230, 159, 0}
\definecolor{wongLightBlue}{RGB}{86, 180, 233}
\definecolor{wongGreen}{RGB}{0, 158, 115}
\definecolor{wongYellow}{RGB}{240, 228, 66}
\definecolor{wongBlue}{RGB}{0, 114, 178}
\definecolor{wongOrange}{RGB}{213, 94, 0}
\definecolor{wongPurple}{RGB}{204, 121, 167}
\definecolor{colUncalibrated}{RGB}{191, 191, 191}
\definecolor{colRecalibrated}{RGB}{197, 214, 231}

\definecolor{bleuTOL}{HTML}{332288}
\definecolor{vertTOL}{HTML}{117733}
\definecolor{vertClairTOL}{HTML}{44AA99}
\definecolor{bleuClairTOL}{HTML}{88CCEE}
\definecolor{sableTOL}{HTML}{DDCC77}
\definecolor{parmeTOL}{HTML}{CC6677}
\definecolor{magentaTOL}{HTML}{AA4499}
\definecolor{roseTOL}{HTML}{882255}

\definecolor{wongPurple}{RGB}{204, 121, 167}
\definecolor{wongLightBlue}{RGB}{86, 180, 233}
\definecolor{gris}{HTML}{A9A9A9}

\DeclareTextFontCommand{\texttt}{\ttfamily\small}

\usepackage{hyperref}
\hypersetup{
    plainpages=false,
    bookmarksopen,
    bookmarksnumbered,
    unicode=true,          
    pdftoolbar=true,        
    pdfmenubar=true,        
    pdffitwindow=false,     
    pdfstartview={FitH},    
    pdftitle={Sequential Conditional Transport on Probabilistic Graphs for Interpretable Counterfactual Fairness},    
    pdfauthor={Fernandes Machado, Agathe and Voitsitska, Iryna and Charpentier, Arthur and Gallic, Ewen},     
    pdfkeywords={fairness, individual fairness, knothe's rearrangement, optimal transport, sequential transport, causality, DAG}, 
    pdfnewwindow=true,      
    colorlinks=true,       
    linkcolor=black,          
    citecolor={wongBlue},        
    filecolor={rose},      
    urlcolor={wongBlue},           
}

\usepackage{doi}
\usepackage{natbib}

\begin{document}



  \author[1]{Agathe {Fernandes Machado}}
  \author[2]{Iryna Voitsitska}
  \author[1,3]{Arthur Charpentier} 
  \author[4,5]{Ewen Gallic}
  \affil[1]{Universit\'e du Qu\'ebec \`a Montr\'eal (UQAM), Canada}
  \affil[2]{Ukrainian Catholic University, Ukraine}
  \affil[3]{Kyoto University, Japan}
  \affil[4]{Aix Marseille Université, CNRS, AMSE, Marseille, France}
   \affil[5]{CNRS - Université de Montréal CRM -- CNRS, Montréal, Canada}
  \title{Sequential Transport for Causal Mediation Analysis}


\maketitle

\abstract{We propose sequential transport (ST), a distributional framework for mediation analysis that combines optimal transport (OT) with a mediator directed acyclic graph (DAG). Instead of relying on cross-world counterfactual assumptions, ST constructs unit-level mediator counterfactuals by minimally transporting each mediator—marginally or conditionally—toward its distribution under an alternative treatment while preserving the causal dependencies encoded by the DAG. For numerical mediators, ST uses monotone (conditional) OT maps based on conditional CDF/quantile estimators; for categorical mediators, it extends naturally via simplex-based transport.
We establish consistency of the estimated transport maps and of the induced unit-level decompositions into mutatis mutandis direct and indirect effects under standard regularity and support conditions. When the treatment is randomized or ignorable (possibly conditional on covariates), these decompositions admit a causal interpretation; otherwise, they provide a principled distributional attribution of differences between groups aligned with the mediator structure. Gaussian examples show that ST recovers classical mediation formulas, while additional simulations confirm good performance in nonlinear and mixed-type settings. An application to the COMPAS dataset illustrates how ST yields deterministic, DAG-consistent counterfactual mediators and a fine-grained mediator-level attribution of disparities.}

\section{Introduction}

In causal reasoning, estimating the effect of a treatment on an outcome often requires accounting for intermediate variables, known as {\em mediators}, that transmit part of the causal influence. Causal mediation analysis provides a formal framework for quantifying these indirect pathways  \citep{robins1992identifiability,vanderweele2015explanation}. Classical applications include evaluating gender-based discrimination in hiring decisions while accounting for differences in qualifications, or assessing medical treatments whose side effects trigger behavioral changes that subsequently affect disease progression \citep{pearl2001direct, pearl2014interpretation}. Mediation allows the total effect of a treatment on an outcome to be decomposed into a {\em direct} effect and an {\em indirect} effect operating through mediators.

This decomposition is  important, in practice. Direct and indirect effects may act in opposite directions, so examining only the total effect can obscure meaningful risks or benefits. 
Moreover, the size of the indirect pathway can guide interventions: strong mediator effects suggest that acting on the mediator may lead to stronger effects than modifying the treatment itself, {\em mutatis mutandis}---that is, ``changing what must be changed'' to remain consistent with the causal structure, rather than a strict {\em ceteris paribus} comparison in which everything else is held fixed \citep{imai2010general}.
Beyond policy analysis and epidemiology, mediation analysis is central in algorithmic fairness, where causal definitions of discrimination rely on isolating pathways from sensitive attributes to predictions \citep{Kusner17}.

Traditional mediation approaches rely on structural equation models (SEM), often in linear form, yielding decompositions such as the Kitagawa--Oaxaca--Blinder representation \citep{kitagawa1955components,oaxaca1973male, blinder1973wage}. Semiparametric or nonparametric alternatives extend these ideas while avoiding strict functional assumptions  \citep{imai2010identification, pearl2009causality}. However, such methods (i) typically target {\em population}-level effects rather than unit-level decompositions, and (ii) may require specifying assumptions on the distributional expressions of a SEM to define mediator counterfactuals, even though the underlying functional forms are generally not identifiable from observational data. In a fully non-parametric setting, deep learning techniques, in particular normalizing flows (NF) have been used to construct counterfactuals and estimate individual-level effects \citep{NEURIPS2020_0987b8b3, javaloy2023causal, zhou2025causalmediationanalysismultiple}. While these techniques offer flexibility in learning the functions of a SEM, they require a sufficient amount of data and may fail to approximate simple functional forms.

Recently, optimal transport (OT) has emerged as a powerful tool for constructing counterfactual representations directly from empirical distributions, without imposing a parametric structural model, and even recovering structural counterfactuals under specific assumptions. OT-based counterfactuals have been used to map individuals from one treatment distribution to another \citep{charpentier2023optimal, de2024transport}, and sequential transport (ST) algorithms have been proposed to align transport operations with the structure of a probabilistic graphical model \citep{machado2025sequential}. 
The core idea is to build, for each untreated unit, a deterministic counterfactual mediator profile that is (i) distributionally aligned with the treated group and (ii) consistent with the mediator DAG, by transporting mediators sequentially along a topological order.
Extensions further allow applying transport maps to categorical mediators by treating them as compositional objects transported on the simplex \citep{machado2025categorical}.
In fairness applications, similar distributional transformations have been used to define counterfactually fair predictors \citep{black2020fliptest, plevcko2021fairadapt}. 

A related line of work investigates the extent to which counterfactuals can be nonparametrically identified when the underlying structural equations are unknown or only partially specified. In particular, \citet{balakrishnan2025conservative} develop conservative inference procedures that characterize {\em sets} of compatible counterfactuals, rather than committing to a single structural model. The approach we consider here is complementary: instead of bounding counterfactuals, we construct a specific distributionally grounded counterfactual, via optimal transport, that satisfies a {\em mutatis mutandis} principle and remains compatible with the observable mediator distributions. Sequential transport may therefore be viewed as selecting one interpretable element from the larger counterfactual ambiguity set considered in this literature.

Finally, recent work has pushed OT-based causal methodologies beyond graph-constrained formulations. \citet{ehyaei2025scot} propose {\em structural causal optimal transport}, which leverages not only the causal DAG but also the {\em structural equations} of a structural causal model to restrict the set of admissible distributions.
They also introduce a {\em relaxed} version where hard causal constraints are replaced by a regularization term, leading to an efficient algorithmic approach, and they provide finite-sample guarantees when the structural equations must be estimated. While their primary focus is ambiguity-set design for distributionally robust optimization, this line of work is complementary to ours: sequential transport exploits a mediator DAG to build conditional transports (monotone in the continuous case, simplex-based in the categorical case), whereas structural causal OT shows how additional functional information (when available) can further refine OT-based causal constraints.

\medskip
\paragraph*{Contributions.}
We integrate these transport-based ideas into causal mediation analysis.\footnote{The replication package, including all data, code, and a replication ebook containing additional results, is available at: \href{https://github.com/fer-agathe/causal_OT}{https://github.com/fer-agathe/causal\_OT}.} For this, we consider the simple but widely used causal ordering
$
A \to \boldsymbol{X} \to Y,
$
where $A$ is a binary treatment, $\boldsymbol{X}$ a (possibly high-dimensional) vector of mediators structured by a directed acyclic graph (DAG), and $Y$ an outcome. Our key idea is to define mediator counterfactuals in a distributional {\em mutatis mutandis} sense: when changing the treatment, each mediator is modified only as necessary to reach the target distribution, while respecting the conditional structure of the joint probability distribution encoded in the DAG.

Technically, we implement this through a {\em sequential transport} scheme that moves mediators one node at a time along a topological ordering \citep{kahn1962topological}. For continuous mediators, each step applies either univariate transport or conditional transport, depending on the parents of the node in the DAG. This yields deterministic, interpretable mediator counterfactuals that remain in-distribution and scale linearly with the number of mediators. For categorical mediators, ST incorporates simplex-based transport to preserve label structure.

Our main contributions are as follows:
\begin{itemize}
\item \textbf{Unit-level mediation via sequential transport.}  
We provide a DAG-aligned OT framework (namely, {\em sequential transport}, ST) to compute individual-level direct and indirect effects, by producing mediator counterfactuals consistent with the causal structure.
    \item \textbf{Theoretical guarantees.}  
    We establish convergence results for the conditional transport maps underlying ST, extending classical multivariate OT theory to sequential, graph-structured settings.
    \item \textbf{Empirical validation.}  
    Gaussian simulations demonstrate the agreement between the ST approach and classical mediation decompositions in linear SEMs, while a fairness application using this transport-based approach on the \texttt{COMPAS} dataset illustrates the interpretability and mediator-level attribution of the indirect effect when using ST.
\end{itemize}

The paper is organized as follows. Section~\ref{sec:2:causalmediation} reviews causal mediation analysis. Section~\ref{sec:3:transportcounterfactuals} introduces OT-based counterfactuals and the {\em mutatis mutandis} perspective. Sequential transport (ST) on DAGs is developed in Section~\ref{sec:general:DAG}. Section~\ref{sec:theory} provides theoretical guarantees for the ST approach. Section~\ref{sec:6:gaussian-sims} presents Gaussian illustrations, and Section~\ref{sec:7:realdata} reports results on real-world data. Section~\ref{sec:discussion} discusses the overall methodology and its possible extensions.

\section{Causal mediation analysis}\label{sec:2:causalmediation}

We use uppercase letters for random variables (e.g., $A,\boldsymbol{X},Y$) and lowercase letters for their realizations (e.g., $a,\boldsymbol{x},y$); boldface denotes vectors, and $i$ indexes units/individuals while $j$ indexes mediators. Let $(A,\boldsymbol{X},Y) \in \mathcal{A} \times \mathcal{X} \times \mathcal{Y}$ be a random vector, where 
$A$ is a binary treatment variable with $\mathcal{A} = \{0,1\}$, 
$\boldsymbol{X} = (X_1,\dots,X_d)$ is a vector of post-treatment mediators and $Y$ is the outcome variable.  We assume $(A,\boldsymbol{X},Y) \sim \mathbb{P}$, where 
$\mathbb{P}$ is a probability measure on the product space $\mathcal{A} \times \mathcal{X} \times \mathcal{Y}$. For each unit $i$, let $(A_i,\boldsymbol{X}_i,Y_i) \sim \mathbb{P}$ denote the $i$-th random observation, and let $(a_i,\boldsymbol{x}_i,y_i)$ denote its realization. We consider the causal ordering
\[
A\to \boldsymbol{X} \to Y  .
\]
This structure appears in many applications and corresponds to the canonical mediation setting described in \citet{robins1992identifiability, pearl2001direct, pearl2014interpretation}. The mediators may themselves be causally related, a dependence that we represent through a DAG $\mathcal{G}$. Let $\mathrm{parents}_{\mathcal{G}}(X_j)$ denote the parents of the random variable $X_j$ in the mediator DAG $\mathcal{G}$. The conditional joint mediator distribution given the treatment value $a$ factorizes as
\[
\mathbb{P}(\boldsymbol{X}\mid A=a)
\;=\;
\prod_{j=1}^d
\mathbb{P}\!\left(X_j \mid \mathrm{parents}_{\mathcal G}(X_j), A=a\right),
\]
consistent with the recursive structure of a DAG, from Markov's property \citep{koller2009probabilistic}. 
Throughout, we distinguish the definition of sequential transport (ST) from its causal interpretation.
ST is defined from the observable conditional laws $\mathbb P(X_j \mid \mathrm{parents}_{\mathcal G}(X_j),A=a)$ (or $\mathbb P(\cdot\mid A=a,C=c)$ if pre-treatment covariates $C$ are used) and from the mediator DAG $\mathcal G$; hence it yields a well-defined distributional transformation regardless of whether $A$ is exogenous.
A causal interpretation requires an additional assumption: either $A$ is randomized, or $A$ is ignorable given observed pre-treatment covariates $C$, so that $\mathbb P(\cdot \mid A=a, C=c)=\mathbb P(\cdot \mid do(A=a), C=c)$.
Under this condition, ST-based decompositions can be interpreted as ({\em mutatis mutandis}) direct and indirect causal effects.

To identify direct and indirect effects, we require the sequential ignorability assumption from \cite{imai2010}, extended to multiple mediators. This entails the absence of unobserved confounding in the treatment-outcome, treatment-mediator, and mediator-outcome relationships, as well as the absence of treatment-induced confounding (see \cite{zhou2025causalmediationanalysismultiple}). When the treatment is randomized, this guarantees no confounding in the treatment-outcome and treatment-mediator relationships, but not necessarily in the other assumptions. In observational settings where $A$ is a sensitive attribute (e.g., race), it is common to treat $A$ as a root node by convention (i.e., we do not model its upstream causes). In this case, identifying causal effects additionally requires assuming that there are no unobserved variables affecting both the outcome and the mediators. When the sensitive attribute is not a root node (e.g., credit score), we must additionally control for confounding in pre-treatment variables. Thus, unless these assumptions can be justified, we do not claim identification of a causal effect of $A$ when using transport plans, in contrast to approaches based on SCMs. Rather, ST should be interpreted as a distributional pathway decomposition of differences between the $A = 0$ and $A = 1$ groups, aligned with $\mathcal{G}$ (see also Section~\ref{sec:general:DAG} for a more detailed discussion).


\subsection{Potential outcomes and effect decomposition}

For each unit $i$, let $\boldsymbol{X}_i(a)$ denote its potential mediator vector under treatment value $a \in \{0,1\}$, and $Y_i(a, \boldsymbol{x})$ its potential outcome under treatment value $a$ and mediator vector value~$\boldsymbol{x}$. The individual total effect of the treatment for unit $i$ is
\[
\tau_i \;=\; Y_i(1, \boldsymbol{X}_i(1)) - Y_i(0, \boldsymbol{X}_i(0)) .
\]

Following \citet{robins1992identifiability} and \citet{pearl2001direct}, this total effect can be decomposed into a {\em natural indirect effect} (NIE),
\[
\delta_i(a) \;=\; Y_i(a, \boldsymbol{X}_i(1)) - Y_i(a, \boldsymbol{X}_i(0)) ,
\]
and a {\em natural direct effect} (NDE),
\[
\zeta_i(a) \;=\; Y_i(1, \boldsymbol{X}_i(a)) - Y_i(0, \boldsymbol{X}_i(a)),
\]
In general (without a no-interaction restriction), the total effect admits the decomposition
$\tau_i = \delta_i(1) + \zeta_i(0)$; under the no-interaction assumption,
$\delta_i(0)=\delta_i(1)$ and $\zeta_i(0)=\zeta_i(1)$, hence $\tau_i=\delta_i+\zeta_i$. These definitions rely on so-called {\em cross-world} (or {\em nested}) counterfactuals, in which potential outcomes from two different hypothetical treatment worlds are combined, as in $Y_i(0,\boldsymbol{X}_i(1))$, as coined in the mediation literature
\citep{andrews2021crossworld,tchetgen2014identification,naimi2015boundless}.

In causal mediation analysis, we are also interested in the average causal mediation effects, defined as
\[
\bar{\tau} \;=\; \mathbb{E}\bigl[ Y(1, \boldsymbol{X}(1)) - Y(0, \boldsymbol{X}(0)) \bigr]  ,
\]
\[
\bar{\delta} \;=\; 
\mathbb{E}\bigl[ Y(0, \boldsymbol{X}(1)) - Y(0, \boldsymbol{X}(0)) \bigr] ,
\]
\[
\bar{\zeta} \;=\;
\mathbb{E}\bigl[ Y(1, \boldsymbol{X}(1)) - Y(0, \boldsymbol{X}(1)) \bigr].
\]


Classical approaches resolve this difficulty by specifying a structural causal model (SCM) for the data-generating process, often using linear SEM. Such models yield tractable decompositions, including Kitagawa--Oaxaca--Blinder formulas, under functional and distributional assumptions \citep{kitagawa1955components,oaxaca1973male, blinder1973wage}. However, individual-level mediator counterfactuals remain difficult to define outside a fully specified SCM, especially in high-dimensional settings or when mediators consist of mixed data types. Some work handles multiple types of a single mediator arranged in a causal graph and enables the identification of individual-level non-parametric effects via Monte Carlo simulation \citep{imai2010general, imai2010identification}. In addition, \citet{zhou2025causalmediationanalysismultiple} extend this work to handle multiple mediators in the graph, using normalizing flows in a non-parametric setting to flexibly learn the structural functions within the SCM, also referred to as {\em causal normalizing flows} \citep{NEURIPS2020_0987b8b3, javaloy2023causal, ijcai2024p907}.






Throughout the paper, we distinguish between the {\em definition} of
transport-based counterfactuals and their {\em causal interpretation}.
The ST construction is {\em purely distributional}:
it is built from the observed conditional laws
$\mathbb{P}(X_j \mid \mathrm{parents}_{\mathcal{G}}(X_j), A=a)$ and from the structure of
the mediator DAG. Hence, ST defines a transported mediator vector as soon as
these conditional distributions are well-defined/estimable, regardless of
whether $A$ is exogenous.
A \emph{causal} interpretation requires additional assumptions. First, when $A$ is a root node in the causal graph (e.g., randomized or ignorable given observed covariates), there is no unobserved confounding between treatment and outcome or between treatment and mediators. In this case, observational and interventional distributions coincide if additional sequential ignorability assumptions from \cite{imai2010, zhou2025causalmediationanalysismultiple}, such as the ignorability of the mediators, are satisfied. Under these conditions, $\mathbb{P}(\cdot \mid A = a) = \mathbb{P}(\cdot \mid do(A = a))$ (or conditionally on covariates). In that case, the resulting ST-based decompositions admit a causal interpretation as {\em mutatis mutandis} direct and indirect effects. In particular, when $A$ is a sensitive attribute such as race in our application, we may represent $A$ as parentless in the graph by convention (that is, we do not model its upstream causes). This is a notational/expositional choice, not an exogeneity assumption. Accordingly, without further assumptions, we do not interpret the resulting decomposition as an identified {\em causal} effect of $A$; rather, it should be read as a {\em distributional} pathway decomposition of between-group differences. Moreover, in observational settings where $A$ may be confounded, ST does not claim to identify causal effects (i.e., it does not in general recover quantities under $do(A=a)$); instead, it provides a principled {\em distributional} decomposition of observed differences between the $A=0$ and $A=1$ groups aligned with the assumed mediator structure.


\subsection{Motivation for distributional counterfactuals}

Our goal is to construct mediator counterfactuals without assuming a parametric SCM.
Instead, we seek a distributional {\em mutatis mutandis} interpretation: when changing the treatment, the mediators should be modified in the minimal way required to place the individual in the appropriate mediator distribution, while respecting the DAG structure. This perspective connects directly to transport-based constructions developed for counterfactual fairness \citep{charpentier2023optimal, de2024transport,machado2025categorical,machado2025sequential}. It provides a principled, nonparametric alternative to specifying a structural model, and motivates the use of OT in the next section.

\section{Optimal transport theory for a {\em mutatis mutandis} counterfactual mechanism}
\label{sec:3:transportcounterfactuals}

Classical causal mediation analysis relies on structural models to define mediator counterfactuals, typically via functional equations linking treatment, mediators, and outcomes. In contrast, transport-based approaches construct counterfactual versions of observed variables directly from their empirical distributions, without committing to a parametric SCM and ensuring in-distribution counterparts. This distributional point of view is particularly appealing when mediators are high-dimensional, 
and when interest lies in {\em individual-level} counterfactuals consistent with observed data geometry.

In this section we present transport–based counterfactual constructions and explain why they provide a natural {\em mutatis mutandis} semantics for mediation analysis. Our goal is to modify an individual's mediator vector only to the extent necessary to represent what their mediators would have been under a different treatment assignment, while preserving the structure of mediator causal relationships encoded in the DAG.


\subsection{Transport-based counterfactuals in the literature}

Hereafter, $\boldsymbol{X}=(X_1,\dots,X_d)$ denotes the (random) mediator vector,
possibly affected by $A$ (directly or through other mediators) and influencing
the outcome of interest $Y$.

The idea of using OT theory to construct realistic counterfactuals has emerged in several areas of causal inference and fairness. In \citet{charpentier2023optimal}, \citet{torous2024otdid} and \citet{de2024transport}, multivariate OT is used to map individuals from the distribution of $\boldsymbol{X}\mid A=0$ to that of $\boldsymbol{X}\mid A=1$, thereby defining nonparametric counterfactual outcomes and features, and without specifying causal relationships between mediators. Building on this viewpoint, \citet{machado2025sequential} develop sequential transport (ST), an algorithmic framework that aligns the transport operation with the causal structure of a probabilistic graphical model. This ensures that transported features remain coherent with respect to conditional dependencies between mediators. For categorical mediators, the simplex-based approach of \citet{machado2025categorical} provides deterministic transport maps that avoid random label switching and preserve interpretable marginal distributions across categories.

Transport-based counterfactuals (including OT and ST) share a common intuition: instead of specifying a full structural causal model, they map each individual to a counterfactual state that is (i) {\em in-distribution} under the target treatment value, and (ii) minimally different from the factual state according to a chosen transport cost. When transport is performed sequentially along a mediator DAG (as in ST), the construction additionally aims to (iii) preserve the DAG-implied conditional semantics by transporting each mediator conditional on (already transported) parent variables. 

\subsection{{\em Mutatis mutandis} semantics}

A central challenge in mediation analysis is the interpretation of cross-world counterfactuals such as $Y(0, \boldsymbol{X}(1))$. These expressions describe what the outcome would have been if the treatment had been set to $0$ ($a\leftarrow0$)
while the mediator had taken whatever value it would have had under treatment value $a=1$ \citep{nguyen2022clarifying}. Such counterfactuals cannot be jointly observed, and their definition typically relies on a SCM with unobserved disturbances.

As previously mentioned, OT theory provides an alternative that is {\em distributional} rather than structural. Instead of requiring a functional model for $\boldsymbol{X}(1)$, we ask:
\begin{quote}
    ``How should the mediator vector of a given individual change when switching the treatment, {\em in a way that makes minimal changes to this vector},
    while ensuring that the modified mediator vector lies in the distribution of mediators under the new treatment value?''
\end{quote}
This principle embodies the classical notion of {\em mutatis mutandis}: only the aspects of the mediator vector that are causally downstream of the treatment should be allowed to change, and only to the extent required to align with the target mediator distribution. OT naturally implements this idea by solving a minimal-cost distributional transportation problem.

\subsection{Transport maps and individual-level counterfactuals}

Let $\nu_0$ and $\nu_1$ denote the distributions of $\boldsymbol{X}$ under $A=0$ and $A=1$. A transport map $T:\mathbb{R}^d\to\mathbb{R}^d$ pushing $\nu_0$ toward $\nu_1$,
i.e. $T_{\#}\nu_0 = \nu_1$, provides an individual-level counterfactual assignment.
For the unit $i$, \[
\boldsymbol{x}_i(1) \;\approx\; T(\boldsymbol{x}_i).
\]
Under suitable regularity conditions, OT yields monotone (when $\nu_0$ and $\nu_1$ are univariate) or cyclical monotone (when their supports lie in $\mathbb{R}^d$) maps solving a minimal deviation principle. In the context of mediators, this gives a nonparametric analogue of the structural relation $\boldsymbol{x}_i(1) = f(a=1, \boldsymbol{u}_i)$, where $\boldsymbol{u}_i$ is unobserved noise for unit $i$ in a structural equation, and $f$ is a vector-valued function that assigns values to the mediators given the exogenous variables $\boldsymbol{U}$ (see \cite{Pearl2016CausalII} for further details on the SCM framework): instead of modeling $f$ explicitly ({\em e.g.}, as with normalizing flows), an OT map constructs the transformation directly from data.

A global multivariate 
OT mapping $T$ is defined as a single map transporting the entire mediator vector $\boldsymbol{x}$ at once, without treating its components individually. While such a map matches the joint target distribution (and thus respects, in a distributional sense, the dependence structure of $\nu_0$ and $\nu_1$), it is agnostic to the {\em direction} of causal relationships encoded by the mediator DAG.
In particular, when $\boldsymbol{X}$ contains causally ordered components, a single multivariate transport may alter the conditional mechanisms (e.g., the laws of $X_j \mid \mathrm{parents}_{\mathcal G}(X_j)$) in ways that are incompatible with the DAG factorization, thereby distorting the mediated causal relationships.
This motivates the sequential construction of Section~\ref{sec:general:DAG}, where mediators are transported one at a time following a topological order of the DAG, so as to preserve the DAG-implied conditional structure of the mediator variables.

\subsection{Transport-induced mutatis mutandis estimands}
\label{sec:transport-estimands}

Let $\nu_a = \mathcal{L}(X \mid A=a)$ denote the mediator distribution in group $A=a$.
More generally than any specific construction, consider a measurable transport mechanism $T_{0\to1}:\mathcal{X}\to\mathcal{X}$ that maps the $A=0$ mediator distribution toward the $A=1$ one (e.g., in the Monge sense $T_{0\to1\#}\nu_0=\nu_1$, or approximately in finite samples).
For an untreated unit, define the transported mediator vector
\[
\boldsymbol{X}^{\dagger}(1) := T_{0\to1}(\boldsymbol{X})\quad \text{under } A=0.
\]
This provides a {\em mutatis mutandis} counterfactual mediator profile: mediators are modified only as needed to align with the target distribution under $A=1$, according to the chosen transport criterion.

\paragraph*{Transport-induced direct and indirect effects.}
Let $\mu_a(\boldsymbol{x}) = \mathbb{E}[Y \mid A=a, \boldsymbol{X}=\boldsymbol{x}]$ for $a\in\{0,1\}$. The transport-induced (indirect, direct, total) effects for the $A=0$ group are defined as
\begin{align}
\delta^{\dagger}
&:= \mathbb{E}\!\left[ \mu_0\!\left(\boldsymbol{X}^{\dagger}(1)\right) - \mu_0(\boldsymbol{X})\,\middle|\, A=0 \right], \label{eq:delta-transport}\\
\zeta^{\dagger}
&:= \mathbb{E}\!\left[ \mu_1\!\left(\boldsymbol{X}^{\dagger}(1)\right) - \mu_0\!\left(\boldsymbol{X}^{\dagger}(1)\right)\,\middle|\, A=0 \right], \label{eq:zeta-transport}\\
\tau^{\dagger}
&:= \delta^{\dagger}+\zeta^{\dagger}
= \mathbb{E}\!\left[ \mu_1\!\left(\boldsymbol{X}^{\dagger}(1)\right) - \mu_0(\boldsymbol{X})\,\middle|\, A=0 \right]. \label{eq:tau-transport}
\end{align}

\paragraph*{Distributional vs.\ causal interpretation.}
The mapping $T_{0\to1}$ and the estimands above are defined purely from observable distributions and thus yield a {\em distributional} decomposition of between-group differences. When $A$ is randomized (or ignorable given pre-treatment covariates), observational and interventional distributions coincide, and the same decomposition admits a causal interpretation as mutatis mutandis direct and indirect effects.


\begin{enumerate}[label=(A\arabic*)]
\setcounter{enumi}{-1}
    \item\label{ass:overlap-outcome} Let $\boldsymbol{X}^{\dagger}(1)=T_{0\to1}(\boldsymbol X)\mid A=0$ and $\mathcal{X}_0=\mathrm{supp}(\boldsymbol  X\mid A=0)$.
We assume that $\mu_0(x)$ is estimable on the transported region in the sense that,
for some small $\eta\ge 0$,
\[
\mathbb{P}\!\left( \boldsymbol{X}^{\dagger}(1)\in \mathcal{X}_0 \,\middle|\, A=0 \right)\ge 1-\eta.
\]
When $\eta=0$, this reduces to $\mathrm{supp}(\boldsymbol{X}^{\dagger}(1))\subseteq \mathcal{X}_0$.
\end{enumerate}


This assumption is a common-support (positivity) condition ensuring that the transported mediators $\boldsymbol X^{\dagger}(1)=T_{0\to1}(\boldsymbol X)\mid A=0$ remain (up to a fraction $\eta$) within the observed support $\mathcal{X}_0=\mathrm{supp}(\boldsymbol X\mid A=0)$, thereby avoiding extrapolation when evaluating the outcome regression $\mu_0(\boldsymbol{x})$ on the transported region.

Although transport-based constructions avoid writing nested (cross-world) counterfactuals of the form $Y(a,\boldsymbol{X}(a'))$ (see, e.g., \citet{vanderweele2014unification,xu2023tutorial}), the plug-in decomposition still evaluates $\mu_0(\cdot)$ at transported mediator profiles.
Assumption~\ref{ass:overlap-outcome} therefore plays the role of a positivity/estimability condition ensuring that outcome regression is learnable on the transported region. In practice, we recommend reporting overlap diagnostics (e.g., propensity score ranges, nearest-neighbor distances from $\boldsymbol{X}^{\dagger}(1)$ to the $A=0$ sample) and, if needed, trimming observations with extreme lack of overlap.

\subsection{Strengths and limitations of transport-based approaches}

Transport-based counterfactuals present several advantages for causal mediation ana\-ly\-sis. First, they provide {\em unit-level} counterfactual mediators without specifying the functional forms within a SCM. Second, they guarantee that counterfactual mediators remain in-distribution. Third, they naturally incorporate causal structures when transport is performed sequentially along a DAG.


At the same time, any transport plan from OT theory constructs counterfactuals that are {\em distributionally} identified: they are defined from the observed group distributions and can be viewed as a guided interpolation between them. In the setting considered in this paper, $A$ is assumed exogenous (a root node), so that conditioning on $A$ coincides with intervention, and the resulting transport-based decompositions admit a causal interpretation.
It should be noted that in observational settings where $A$ may be confounded, transport-based counterfactuals remain well-defined as distributional objects, but interpreting them causally would require additional assumptions (e.g., ignorability given covariates) or appropriate adjustment.

The next section formalizes this idea by providing the detailed ST mechanism that respects mediator causal relationships encoded in a DAG. This ST approach extends naturally to categorical variables, and yields well-defined individual-level direct and indirect effects.

\section{Sequential transport on a mediator DAG}
\label{sec:general:DAG}

Section~\ref{sec:transport-estimands} introduced transport--based counterfactual mediators through a {\em generic} map $T_{0\to1}:\mathcal X\to\mathcal X$ that pushes $\nu_0:=\mathcal L(\boldsymbol X\mid A=0)$ toward $\nu_1:=\mathcal L(\boldsymbol X\mid A=1)$.
A direct multivariate construction of $T_{0\to1}$, however, need not respect the internal structure of $\boldsymbol X$: when mediators obey a directed acyclic graph (DAG), a global transformation may distort the conditional relationships encoded by the graph and produce counterfactual profiles that are misaligned with the intended parent--child semantics.

To obtain counterfactual mediators that are both distributionally valid and {\em graph-aligned}, we instantiate $T_{0\to1}$ through a {\em sequential} transport: we update the mediators one coordinate at a time, in a topological order of the mediator DAG, and condition each update on the already-transported parent values. This yields a deterministic transported profile $\boldsymbol x^\star_{1,i}$ for each observation $\boldsymbol x_i$ in group $A=0$.

Our construction can be viewed as a graphical analogue of the Knothe--Rosenblatt (KR) rearrangement, which itself is rooted in the Rosenblatt factorization (chain rule) \citep{rosenblatt1952remarks} and was formalized by \citet{knothe1957contributions}, with later probabilistic graphical formulations discussed in \citet{lauritzen2019lectures, cheridito2023optimal}. Closely related ideas have recently been used for ``{counterfactual fairness}'' in \citet{machado2025categorical, machado2025sequential}. Here we adapt and formalize them in our mediation setting, covering both numerical and categorical mediators.

\subsection{Mediator DAG and notation}
\label{sec:dag-notation}

Let $A\in\{0,1\}$ denote the treatment and $\boldsymbol{X}=(X_1,\dots,X_d)$ the mediator vector.
Let $\mathcal G$ be a mediator DAG whose nodes are $\{X_1,\dots,X_d\}$ and whose directed edges encode dependencies between mediators downstream of $A$.
For each mediator $X_j$, write $\boldsymbol Z_j=\mathrm{parents}_{\mathcal G}(X_j)$ for its parent set.
Because $\mathcal G$ is acyclic, there exists a topological ordering $X_{\pi(1)},\dots,X_{\pi(d)}$ such that all parents of $X_{\pi(j)}$ appear among $\{X_{\pi(1)},\dots,X_{\pi(j-1)}\}$ (see \citet{kahn1962topological} and for practical implementation, several algorithms are discussed in Section 20.4 in \citet{cormen2022introduction}).

For $a\in\{0,1\}$ and each $j$, denote by
\[
\mathbb P_{a,j}(\cdot \mid \boldsymbol z)
~ \text{the conditional law of}~
X_j \mid (\boldsymbol Z_j=\boldsymbol z, A=a).
\]
In the data, we observe $n_a$ units in group $A=a$, with index set $\mathcal I_a=\{i:A_i=a\}$ and mediator vectors $\boldsymbol x_i=(x_{i1},\dots,x_{id})$.
For each $j$, we form the empirical sample of child--parent pairs
\[
\chi_{a,j}
=
\left\{ \bigl(x_{ij},\,\boldsymbol x_{i,\boldsymbol Z_j}\bigr) : i\in\mathcal I_a \right\},
\]
which summarizes information about $\mathbb P_{a,j}(\cdot\mid \boldsymbol z)$.
When $\boldsymbol Z_j=\emptyset$, this reduces to the marginal sample of $X_j$ in group $A=a$.

\subsection{Sequential transport}
\label{sec:st-population}

\paragraph*{Population-based sequential transport}

Sequential Transport (ST) is defined by conditional transport maps that are composed along the ordering $\pi$. The key point is that, once parents have been transported, the {\em target} conditional law for a child should be evaluated at the {\em transported} parent values.
Accordingly, it is convenient to define a two-argument conditional transport: for each mediator $j$, let $T_j(\cdot\mid \boldsymbol z_0,\boldsymbol z_1)$ be a measurable map such that, for each $(\boldsymbol z_0,\boldsymbol z_1)$,
\[
\bigl(T_j(\cdot\mid \boldsymbol z_0,\boldsymbol z_1)\bigr)_{\#}\,
\mathbb P_{0,j}(\cdot\mid \boldsymbol z_0)
=
\mathbb P_{1,j}(\cdot\mid \boldsymbol z_1).
\]
Here $\boldsymbol z_0$ indexes the source conditional (from group $A=0$), whereas $\boldsymbol z_1$ indexes the target conditional (in group $A=1$). When $\boldsymbol Z_j=\emptyset$, we simply write $T_j$ for the
marginal transport.

For an individual $i$ with $A_i=0$ and factual mediator vector $\boldsymbol x_{0,i}$, the population ST counterfactual $\boldsymbol x^\dagger_{1,i}$ is constructed recursively:
\[
\begin{cases}
    x^\dagger_{1,i,\pi(1)} = T_{\pi(1)}\!\left(x_{0,i,\pi(1)}\right),
\\
x^\dagger_{1,i,\pi(j)} =
T_{\pi(j)}\!\left(
x_{0,i,\pi(j)} \,\middle|\,
\boldsymbol x_{0,i,\mathrm{parents}_{\mathcal{G}}(X_{\pi(j)})},\,
\boldsymbol x^\dagger_{1,i,\mathrm{parents}_{\mathcal{G}}(X_{\pi(j)})}
\right),
\quad j=2,\dots,d,
\end{cases}
\]
and we set $\boldsymbol x^\dagger_{1,i}=(x^\dagger_{1,i,\pi(1)},\dots,x^\dagger_{1,i,\pi(d)})$.

Intuitively, ST modifies each mediator just enough to match the appropriate conditional distribution in
group $A=1$, while conditioning on the already-transported parents so as to remain aligned with the DAG.
This contrasts with applying a single multivariate transport map to $\boldsymbol X$, which generally does not
preserve the intended parent--child conditioning relationships encoded by $\mathcal G$.

\paragraph*{Sample-based sequential transport}

In practice, we replace $\{T_j\}_{j=1}^d$ by estimates $\{\widehat T_j\}_{j=1}^d$ learned from the samples $\chi_{0,j}$ and $\chi_{1,j}$, yielding the estimated transported profile $\boldsymbol x^\star_{1,i}$ via the same recursion with $\widehat T_j$ in place of $T_j$.
The next subsections describe numerical and categorical constructions of $\widehat T_j$ and establish consistency results for the resulting sequentially transported mediators.

\subsection{Transport for numerical mediators}

When $X_j$ is numerical, we estimate a univariate or conditional transport map depending on whether $X_j$ has parents beyond the treatment variable.

\paragraph*{No parents.}  
If $\mathrm{parents}_{\mathcal{G}}(X_j)=\emptyset$, the mediator depends only on the treatment. In this case, we apply the classical univariate OT map based on smoothed cumulative distribution function (CDF) and quantile function:
\[
\forall x \in \mathcal{X}_j, \enspace T_{j}(x) = Q_{1,j}\!\circ F_{0,j}(x),
\]
where $F_{0,j}$ is a smoothed empirical CDF under $A=0$, and $Q_{1,j}$ a smoothed quantile function under $A=1$, i.e., $F_{0,j}(x) := \mathbb{P}(X_j \le x \mid A=0)$ while $Q_{1,j}(u) := \inf\{x\in\mathbb{R} : F_{1,j}(x) \ge u\}$, for $u\in(0,1)$, where $F_{1,j}(x):=\mathbb{P}(X_j\le x\mid A=1)$.

\paragraph*{Parents present.}

If $\mathrm{parents}_{\mathcal G}(X_j)\neq\emptyset$, we transport $X_j$ conditionally on its parents.
Let $\boldsymbol Z_j=\mathrm{parents}_{\mathcal G}(X_j)$. For each $\boldsymbol z_j$, we seek a map
$T_j(\cdot\mid \boldsymbol z_j)$ such that
\[
\bigl(T_j(\cdot\mid \boldsymbol z_j)\bigr)_{\#}\, \mathbb{P}\!\left(X_j \mid \boldsymbol Z_j=\boldsymbol z_j, A=0\right)
=
\mathbb{P}\!\left(X_j \mid \boldsymbol Z_j=T_{\boldsymbol Z_j}(\boldsymbol z_j), A=1\right).
\]

We implement this using the kernel-based conditional transport estimator of \citet{machado2025sequential}, which approximates the conditional CDF under $A=0$ and the corresponding conditional quantile under $A=1$. The mapping thus generalizes the Rosenblatt (or Knothe--Rosenblatt) construction to a conditional and treatment-indexed setting.

\subsection{Transport for categorical mediators}

When $X_j$ is categorical with $K$ levels, we write $X_j\in [K]=\{1,\dots,K\}$ and represent its conditional law given its parents as a probability vector on the $(K-1)$-simplex $\mathcal S_{K-1}=\{\boldsymbol p\in\mathbb R_+^K:\sum_{k=1}^K p_k=1\}$. Specifically, for a parent value $\boldsymbol z$ and group $A=a$, let
\[
\boldsymbol p_{a,j}(\boldsymbol z)
=
\bigl(\mathbb P(X_j=k\mid \boldsymbol Z_j=\boldsymbol z, A=a)\bigr)_{k\in[K]}
\in \mathcal S_{K-1},
\]
which we estimate in practice (e.g., via multinomial/softmax regression or local smoothing), yielding simplex-valued fitted probabilities for each unit.

We then transport {\em distributions on the simplex} rather than labels directly.
Following \citet{machado2025categorical}, we construct a deterministic OT map
$T_j:\mathcal S_{K-1}\to \mathcal S_{K-1}$ that pushes the source fitted probabilities
toward the target ones. For a unit with $A=0$ and parents $\boldsymbol z$, this produces a
transported probability vector $\tilde{\boldsymbol p}_j = T_j(\boldsymbol p_{0,j}(\boldsymbol z))$.

To obtain a categorical counterfactual in $[K]$, we finally map $\tilde{\boldsymbol p}_j$ to a vertex
of the simplex, i.e., to a unit vector $\boldsymbol e_k$ (the $k$th canonical basis vector)
corresponding to category $k$. A simple deterministic choice is argmax rounding,
$\hat x_j(1)=\arg\max_{k\in[K]} (\tilde{\boldsymbol p}_j)_k$.
However, argmax rounding does not in general preserve the target category proportions.
When we wish the induced categorical counterfactuals to match the target marginals in group $A=1$,
we instead perform a global (semi-discrete) OT allocation of the transported vectors
$\tilde{\boldsymbol p}_j$ onto the vertices $\{\boldsymbol e_1,\dots,\boldsymbol e_K\}$ so that the
resulting one-hot assignments have empirical frequencies equal (or close) to the target
category marginals; see Appendix~\ref{app:algorithms} and Algorithm~\ref{alg:split}.

When parents are present, the above steps are applied conditionally: we compute
$\boldsymbol p_{a,j}(\boldsymbol z)$ at the relevant parent values, transport using the already-transported
parents (as in the continuous case), and apply the rounding/allocation within the corresponding
conditional target distribution.

\subsection{Resulting counterfactual mediator vector}

Applying the above steps to each mediator in topological order yields a fully transported counterfactual vector
\[
\boldsymbol{x}_{1,i}^\star
    = \bigl( x^\star_{1,\pi(1)}, \dots, x^\star_{1,\pi(d)} \bigr),
\]
which (i) belongs to the support of the $A=1$ mediator distribution, (ii) respects the DAG structure between mediator variables, and (iii) differs ``minimally'' from $\boldsymbol{x}_{0,i}$ in transport cost. The mapping is deterministic and therefore provides a stable basis for individual-level analysis.

Algorithm~\ref{alg:0} summarizes the ST procedure; detailed implementations for
numerical and categorical mediators are provided in Appendix~\ref{app:algorithms}. We assume Assumption~\ref{ass:overlap-outcome} so that evaluating 
$\widehat{\mu}_0(\boldsymbol{x}_{1,i}^\star)$ does not rely on extrapolation.

\begin{algorithm}[t!]
    \caption{Sequential transport on causal graph $\mathcal{G}$ for untreated individual $\boldsymbol{x}_0$.}\label{alg:0}
    \begin{algorithmic}
        \Require graph $\mathcal{G}$ on $(a,\boldsymbol{x})$
        \State $(a,\boldsymbol{v})\gets$ the topological ordering of vertices (DFS)
        \State $T_j\gets\text{identity}$, $j\in\boldsymbol{v}$ and $\boldsymbol{x}^\star_1\gets\boldsymbol{x}_0$
        \Require model $\widehat{\mu}_0(\cdot)\gets\{y,\boldsymbol{x}\}$, on untreated individuals (group $A=0$)
        \Require model $\widehat{\mu}_1(\cdot)\gets\{y,\boldsymbol{x}\}$, on treated individuals (group $A=1$)
        \Statex \textbf{Assume} Assumption~\ref{ass:overlap-outcome}  when evaluating $\widehat{\nu}_0$ at transported mediator profiles
        \Require probabilistic models $\widehat{p}_j:\mathcal{X}\to\mathcal{S}_{d_j-1}$, for all categorical variables $x_{j}$, $\widehat{p}_j(\cdot)\gets\{x_{j},\boldsymbol{x}_{-j}\}$, 
        \Require {optimal categorical splitting function $C_j$, for all categorical variables $x_{j}$, Algorithm \ref{alg:split}}
        \Require smoothing parameters $\boldsymbol{h}$
        \For{$j\in \boldsymbol{v}$} 
            \If{$\text{parents}_{\mathcal{G}}(x_j)=\{a\}$}
                \State $\boldsymbol{w}_0\gets \boldsymbol{1}_{n_0}/n_0$ and $\boldsymbol{w}_1\gets \boldsymbol{1}_{n_1}/n_1$
            \Else
                \State $\boldsymbol{w}_{0,i}\gets \displaystyle K_{h_j}(\|\boldsymbol{x}_{0,\text{parents}_{\mathcal{G}}(x_j),i}-\boldsymbol{x}_{0,\text{parents}_{\mathcal{G}}(x_j)}\|)$
                \State $\boldsymbol{w}_{1,i}\gets \displaystyle K_{h_j}(\|\boldsymbol{x}_{1,\text{parents}_{\mathcal{G}}(x_j),i}-\boldsymbol{x}^\star_{1,\text{parents}_{\mathcal{G}}(x_j)}\|)$
            \EndIf
            \If{$x_j$ numeric}
                \State $\chi_{0,j}\gets\{x_{0,j,1},\cdots,x_{0,j,n_0}\}\in\mathbb{R}^{n_0}$ 
                \State $\chi_{1,j}\gets\{x_{1,j,1},\cdots,x_{1,j,n_1}\}\in\mathbb{R}^{n_1}$ 
                \State $x_{1,j}^\star\gets $ Algorithm \ref{alg:1} $(x_{0,j},\boldsymbol{w}_0,\boldsymbol{w}_1,\chi_{0,j},\chi_{1,j})$
            \ElsIf{$x_j$ categorical}
                \State $\chi_{0,j}\gets\{\widehat{p}_j(x_{0,j,1}),\cdots,\widehat{p}_j(x_{0,j,n_0)}\}\in\mathcal{S}_{d_j-1}^{n_0}$ 
                \State $\chi_{1,j}\gets\{\widehat{p}_j(x_{1,j,1}),\cdots,\widehat{p}_j(x_{1,j,n_1)}\}\in\mathcal{S}_{d_j-1}^{n_1}$ 
                \State $x_{1,j}^\star\gets $ Algorithm \ref{alg:2} $(x_{0,j},\boldsymbol{w}_0,\boldsymbol{w}_1,\chi_{0,j},\chi_{1,j},C_j)$
            \EndIf
            \State ${x}_{1,j}^\star\gets T_{j}(\boldsymbol{x}_{0})$
        \EndFor\\
        \Return $\boldsymbol{x}^\star_1$, counterfactual of $\boldsymbol{x}_0$ \\
        \Return indirect effect $\delta\gets \widehat{\mu}_0(\boldsymbol{x}^\star_1)-\widehat{\mu}_0(\boldsymbol{x}_0)$ \\
        \Return direct effect $\zeta \gets \widehat{\mu}_1(\boldsymbol{x}^\star_1)-\widehat{\mu}_0(\boldsymbol{x}^\star_1)$ \\
        \Return total effect $\tau\gets \delta+\zeta$ 
    \end{algorithmic}
\end{algorithm}

\subsection{Direct and indirect effects via sequential transport}\label{sec:st-effects-decomposition}

Given a predictive model $\widehat{\mu}_a(\cdot)$ for the outcome under treatment $A=a$, modeling $\mu_a(\boldsymbol{x}) = \mathbb{E}[Y \mid A=a,\boldsymbol{x}]$, ST yields unit-level decompositions of the total effect:
\[
\begin{cases}
\delta_i = \widehat{\mu}_0(\boldsymbol{x}_{1,i}^\star)
          - \widehat{\mu}_0(\boldsymbol{x}_{0,i}), & \text{(indirect)}\\
\zeta_i = \widehat{\mu}_1(\boldsymbol{x}_{1,i}^\star)
          - \widehat{\mu}_0(\boldsymbol{x}_{1,i}^\star),& \text{(direct)} \\
\tau_i = \delta_i + \zeta_i.
\end{cases}
\]
The indirect effect $\delta_i$ isolates the contribution of the mediator transformation, while the direct effect $\zeta_i$ captures the change in treatment holding the mediator counterfactual fixed in its transported state. Both quantities exist at the individual level, require no parametric SCM, and respect the causal ordering encoded in $\mathcal{G}$.

These expressions parallel the natural direct and indirect effects by replacing the cross-world structural counterfactuals $\boldsymbol{x}(1)$ with ST-constructed counterfactuals $\boldsymbol{x}_{1,i}^\star$. Although $\boldsymbol{x}_{1,i}^\star$ remains within the observed distribution of the data, an individual characterized by $(A = 0, \boldsymbol{x}_{1,i}^\star)$ is inherently cross-world and is never observed in practice, since it combines information from both the $A = 0$ and $A = 1$ worlds.


Section~\ref{sec:6:gaussian-sims} illustrates the construction of the ST-based counterfactuals in Gaussian settings, while Section~\ref{sec:7:realdata} applies it to a real-world fairness problem.

\section{Theoretical guarantees for sequential transport}\label{sec:theory}

ST constructs mediator counterfactuals by composing a series of univariate or conditional transport maps along a topological ordering of the mediator DAG. This section establishes statistical guarantees for these maps and, by extension, for the counterfactual mediators and effect decompositions they induce. The results extend classical convergence analyses for monotone OT to a conditional and graph-structured setting.

Throughout, we work under assumptions compatible with the exogeneity of the treatment $A$ (a root node in the DAG) and with the absolute continuity or discrete-support structure of the mediators discussed previously. All proofs are deferred to Appendix~\ref{app:proofs}.

\subsection{Setting and assumptions}

Let $X_j$ be a mediator in the DAG with parent set $\boldsymbol{Z}_j = \mathrm{parents}_{\mathcal{G}}(X_j)$. We denote by
\[
F_{a,j}(\cdot \mid \boldsymbol{z}), \qquad Q_{a,j}(\cdot \mid \boldsymbol{z})
\]
the conditional cumulative distribution and quantile functions of $X_j$ given $\boldsymbol{Z}_j = \boldsymbol{z}$ under treatment $A=a$. 
ST uses kernel-based estimators $\widehat{F}_{0,j}(\cdot \mid \boldsymbol{z})$ and
$\widehat{Q}_{1,j}(\cdot \mid \boldsymbol{z})$ to form the estimated conditional transport.
To make explicit the sequential nature of ST, we distinguish the parent values used for the
source conditional CDF and the target conditional quantile. For $(\boldsymbol{z}_0,\boldsymbol{z}_1)$,
define
\[
\widehat{T}_{j}(x \mid \boldsymbol{z}_0,\boldsymbol{z}_1)
=
\widehat{Q}_{1,j}\!\left( \widehat{F}_{0,j}(x \mid \boldsymbol{z}_0) \mid \boldsymbol{z}_1 \right).
\]
In the sequential procedure, we take $\boldsymbol{z}_0$ to be the observed parent values in group $A=0$
and $\boldsymbol{z}_1$ to be their already-transported counterparts, i.e.,
$\boldsymbol{z}_1=\widehat{T}_{\boldsymbol{Z}_j}(\boldsymbol{z}_0)$, where
$\widehat{T}_{\boldsymbol{Z}_j}$ denotes the vector of transport maps for the components of
$\boldsymbol{Z}_j$ (in the chosen topological order). Equivalently, writing
$\widehat{T}(\boldsymbol{z})$ for the transported parent vector, we have the shorthand
\[
\widehat{T}_{j}(x \mid \boldsymbol{z})
=
\widehat{Q}_{1,j}\!\left( \widehat{F}_{0,j}(x \mid \boldsymbol{z}) \mid \widehat{T}(\boldsymbol{z}) \right).
\]
as in \cite{machado2025sequential}. Similarly, categorical mediators are handled via simplex-based transport, leading to estimators $\widehat{T}_j : \mathcal{S}_{K} \to \mathcal{S}_{K}$ following \citet{machado2025categorical}.

Regularity assumptions are necessary here. Namely, we assume that \ref{A1} conditional mediator distributions admit densities or finite supports; \ref{A2} conditional CDFs are strictly increasing; \ref{A3} kernel estimators of conditional distributions converge uniformly; \ref{A4} the mediator graph is a known DAG. Formally,

\begin{enumerate}[label=(A\arabic*)]
    \item\label{A1}
    For each $j$, the conditional law of $X_j \mid \boldsymbol{Z}_j=\boldsymbol{z}, A=a$ has a density with respect to the Lebesgue measure (numerical case) or a probability mass function on a finite set (categorical case).
    \item\label{A2}
    The conditional CDFs $F_{a,j}(\cdot \mid \boldsymbol{z})$ are strictly increasing on the support of $X_j \mid \boldsymbol{Z}_j=\boldsymbol{z}, A=a$.
    \item\label{A3}
    Kernel estimators for the conditional distribution satisfy uniform convergence:
    \[
    \sup_{x,\boldsymbol{z}} \bigl| \widehat{F}_{0,j}(x \mid \boldsymbol{z})-F_{0,j}(x \mid \boldsymbol{z}) \bigr|
        \xrightarrow[]{p} 0,
    \text{ and }
    \sup_{u,\boldsymbol{z}} \bigl| \widehat{Q}_{1,j}(u \mid \boldsymbol{z})-Q_{1,j}(u \mid \boldsymbol{z}) \bigr|
        \xrightarrow[]{p} 0.
    \]
    \item\label{A4}
    The DAG $\mathcal{G}$ is known and acyclic, allowing a valid topological ordering.
\end{enumerate}

Assumption~\ref{A3} is standard in nonparametric kernel estimation and holds under mild smoothness and bandwidth conditions; see, e.g., \citet{einmahl_mason_2005} or \citet[Ch.~6]{li2007nonparametric}.
Moreover, under a conditional density bounded away from zero around the relevant quantiles, uniform convergence of $\widehat F_{a,j}$ implies uniform convergence of the corresponding conditional quantile estimator by standard inversion arguments (e.g., \citet[Thm.~21.2]{vandervaart1998}). 
The other assumptions are standard in transport literature, and facilitate the consistency of ST mappings.

For numeric mediator, we will also require regularity in the transported covariate argument \ref{A6}, 

\begin{enumerate}[label=(A\arabic*)]
\setcounter{enumi}{4}
\item\label{A6}
There exists $\varepsilon\in(0,1/2)$ such that, for any compact $\mathcal K_j$ contained in the interior of
$\mathrm{supp}(\boldsymbol Z_j)$, the map
$\boldsymbol z\mapsto Q_{1,j}(u\mid \boldsymbol z)$ is uniformly continuous on $\mathcal K_j$ uniformly over
$u\in[\varepsilon,1-\varepsilon]$; for instance, a Lipschitz condition:
\[
\sup_{u\in[\varepsilon,1-\varepsilon]}\bigl|Q_{1,j}(u\mid \boldsymbol z)-Q_{1,j}(u\mid \boldsymbol z')\bigr|
\le L_j\|\boldsymbol z-\boldsymbol z'\|\qquad\forall\,\boldsymbol z,\boldsymbol z'\in\mathcal K_j.
\]
\end{enumerate}

Assumption~\ref{A6} is also mild and standard in quantile regression: it holds, for instance, if the conditional density is bounded away from zero around the relevant quantiles and the conditional CDF varies smoothly in $\boldsymbol z_1$, ensuring stability of inversion; see, e.g., \citet[Sec.~2.2]{koenker2005quantile} and related regularity conditions in \citet{kato2012estimation}.

We use the following standard stability condition for discrete OT on the simplex, with \ref{C1}, continuity/stability of the simplex transport operator; and \ref{C2}, consistency of conditional probability estimators.

\begin{enumerate}[label=(A\arabic*)]
\setcounter{enumi}{5}
\item\label{C1}
For each $j$ and $(\boldsymbol z_0,\boldsymbol z_1)$ in the interior of $\mathrm{supp}(\boldsymbol Z_j)$, the optimal simplex transport
operator $T_j(\cdot\mid \boldsymbol z_0,\boldsymbol z_1)$ is continuous at
$\bigl(p_{0,j}(\cdot\mid \boldsymbol z_0),p_{1,j}(\cdot\mid \boldsymbol z_1)\bigr)$; e.g., the optimal plan/map is unique
(or a deterministic continuous selection is fixed).
\item\label{C2}
For $a\in\{0,1\}$,
$\displaystyle\sup_{k,\boldsymbol z}\bigl|\widehat p_{a,j}(k\mid \boldsymbol z)-p_{a,j}(k\mid \boldsymbol z)\bigr|\xrightarrow{p}0$.
\end{enumerate}

Assumption~\ref{C1} is a stability requirement for the discrete OT solution on the simplex: for fixed costs, the simplex transport is the solution of a linear program whose optimal solution mapping is (piecewise) affine and is continuous at points where the optimal plan/map is unique and nondegenerate. This type of local stability is classical in sensitivity/perturbation analysis of linear programs \citep{BonnansShapiro1998,BonnansShapiro2000}. When the OT problem admits multiple optimal plans (a measure-zero ``degeneracy'' situation under generic costs), continuity may fail for an arbitrary selection; in practice, one may either fix a deterministic continuous selection (as done in \citet{machado2025categorical}) or introduce a small strongly convex regularization (e.g., entropic) which enforces uniqueness and smooth dependence on the marginals as discussed in Appendix~\ref{app:pen:match}.

Assumption~\ref{C2} concerns estimation of the conditional probability vectors $p_{a,j}(\cdot\mid \boldsymbol z)$ on a finite support. In simple parametric models such as the multinomial logit, $p_{a,j}(k\mid \boldsymbol z)=\pi_k(\boldsymbol z;\beta_a)$ with a smooth link in $(\boldsymbol z,\beta_a)$, the MLE $\widehat\beta_a$ is consistent under standard regularity/identifiability conditions for generalized linear models (e.g., no separation and finite Fisher information), and hence $\widehat p_{a,j}(k\mid \boldsymbol z)=\pi_k(\boldsymbol z;\widehat\beta_a)$ converges to $p_{a,j}(k\mid \boldsymbol z)$.
Moreover, on compact subsets of the covariate space (bounded regressors), smoothness of $\pi_k$ implies that this convergence can be strengthened to uniform convergence in $\boldsymbol z$ (a standard continuous-mapping/Glivenko--Cantelli argument for parametric classes); see, e.g.,  \citet{FahrmeirKaufmann1985} and textbook treatments of multinomial logit such as \citet{Agresti2002}.

Finally, note that multinomial logit probabilities are strictly positive for finite parameters, so the conditional laws typically lie in the interior of the simplex; however, probabilities can become arbitrarily close to $0$ if linear predictors diverge, which may lead to numerical instability of the discrete OT solution. A common practical safeguard is to restrict attention to regions where $\min_k p_{a,j}(k\mid \boldsymbol z)\ge c>0$ (or to apply mild regularization), ensuring the transported assignment is stable.

\subsection{Pointwise consistency of conditional transport maps}\label{sec:5:2}

\paragraph*{Root mediator}

\begin{lme}[Consistency for numeric root mediators]
\label{lem:base-root}
Let $j$ be such that $\mathrm{parents}_{\mathcal{G}}(X_j)=\emptyset$ (so $\boldsymbol Z_j$ is empty).
Under Assumptions~\ref{A1}--\ref{A3}, the estimated marginal transport map is consistent:
in the numerical case,
\[
\sup_{x\in\mathcal X_j}\bigl|\widehat T_j(x)-T_j(x)\bigr|\xrightarrow[]{p}0,
\text{ where }
\begin{cases}
    T_j(x)=Q_{1,j}(F_{0,j}(x)),
\\
\widehat T_j(x)=\widehat Q_{1,j}(\widehat F_{0,j}(x)),
\end{cases}
\]
for any compact $\mathcal X_j$ contained in the interior of $\mathrm{supp}(X_j\mid A=0)$.
\end{lme}

In the categorical case (finite support), the plug-in simplex transport map is consistent in the same sense.

\paragraph*{Inductive intuition (sequential plug-in).}

Because sequential transport (ST) follows a topological ordering of the mediator DAG, each step $j$ only depends on quantities associated with nodes $k<j$ (i.e., already-transported ancestors/parents).
For root mediators (no parents), consistency was discussed in the previous sub-section.
Moving to a general node $j$, the only additional difficulty is that the conditional transport map
$\widehat T_j(\cdot\mid \boldsymbol z_0,\boldsymbol z_1)$ is evaluated at a {\em random} covariate value
$\widehat{\boldsymbol z}_1=\widehat T_{\boldsymbol Z_j}(\boldsymbol z_0)$ produced by earlier stages.
This is handled by a standard plug-in decomposition: (i) consistency of $\widehat T_j(\cdot\mid \boldsymbol z_0,\boldsymbol z_1)$ at fixed
$(\boldsymbol z_0,\boldsymbol z_1)$, and (ii) regularity (e.g., continuity/Lipschitz) of the population map in $\boldsymbol z_1$.
When working at step $j$, we thus assume that all previous consistency statements (for $k< j$) have already been established, i.e.,

\begin{enumerate}[label=($\text{H}(j)$)]
    \item\label{A7} At step $j$, for all $k< j$, and any $\boldsymbol{z}_0$ in the interior of $\mathrm{supp}(Z_k)$
    \[
\widehat T_{\boldsymbol Z_k}(\boldsymbol z_0)\xrightarrow[]{p} T_{\boldsymbol Z_k}(\boldsymbol z_0).
\]
\end{enumerate}

\ref{A7} is an induction hypothesis used for intermediate plug-in results; it is proved recursively in Theorem~\ref{thm:seq:consistency}.
Now, we can discuss consistency of $\widehat{T}_j(x\mid \boldsymbol{z}_0)$ when $X_j$ is either numerical or categorical.

\paragraph*{Numerical mediator}

We begin with the case where $X_j$ is a numerical mediator.

\begin{prop}[Consistency of conditional numerical transport]
\label{prop:cv:Tx}
Under Assumptions~\ref{A1}--\ref{A3}, for any mediator $X_j$ and any $(x,\boldsymbol{z}_0,\boldsymbol{z}_1)$
in the interior of the support of $(X_j,\boldsymbol{Z}_j)$,
\[
\widehat{T}_{j}(x \mid \boldsymbol{z}_0,\boldsymbol{z}_1)
    \xrightarrow[]{p}
    T_{j}(x \mid \boldsymbol{z}_0,\boldsymbol{z}_1)
    := Q_{1,j}\!\left( F_{0,j}(x \mid \boldsymbol{z}_0 ) \mid \boldsymbol{z}_1 \right).
\]
\end{prop}

\begin{cor}[Sequential (parent-transported) version]
\label{cor:cv:Tx-seq}
In the sequential procedure, we set $\boldsymbol{z}_1 = T_{\boldsymbol{Z}_j}(\boldsymbol{z}_0)$,
so that
\[
\widehat{T}_{j}\!\left(x \mid \boldsymbol{z}_0,{T}_{\boldsymbol{Z}_j}(\boldsymbol{z}_0)\right)
\xrightarrow[]{p}
T_{j}\!\left(x \mid \boldsymbol{z}_0, T_{\boldsymbol{Z}_j}(\boldsymbol{z}_0)\right).
\]
\end{cor}

\begin{prop}[Plug-in along estimated transported parents]\label{prop:5:2}
Assume \ref{A1}--\ref{A6} and \ref{A7}. Fix $\boldsymbol z_0$ in the interior of $\mathrm{supp}(Z_j)$ and set $\boldsymbol z_1=T_{\boldsymbol Z_j}(\boldsymbol z_0)$ and $\hat{\boldsymbol{z}}_1=\widehat{T}_{\boldsymbol Z_j}(\boldsymbol z_0)$. Moreover, if \ref{A6} holds for some $\varepsilon>0$ and if $F_{0,j}(x\mid \boldsymbol z_0)\in[\varepsilon,1-\varepsilon]$, by consistency of the sequential transport estimator for $\boldsymbol Z_j$,
\[
\widehat{T}_{j}\!\left(x \mid \boldsymbol{z}_0\right)
:=\widehat{T}_j\!\left(x\mid \boldsymbol z_0,\widehat{T}_{\boldsymbol Z_j}(\boldsymbol z_0) \right)
\xrightarrow{p}T_j\!\left(x\mid \boldsymbol{z}_0,{T}_{\boldsymbol Z_j}(\boldsymbol z_0)\right)=:{T}_{j}\!\left(x \mid \boldsymbol{z}_0\right).
\]
\end{prop}

This result generalizes classical consistency of univariate monotone transport to a conditional setting. The quality of the approximation propagates through the sequence of transport operations, yielding consistent mediator counterfactuals for ST.

\paragraph*{Categorical mediator}

We now turn to categorical mediators. Let $p_{a,j}(\cdot \mid \boldsymbol{z})$ denote the conditional multinomial probability vector of $X_j$ under $A = a$, given the parent values $\boldsymbol{z}$.

\begin{prop}[Pointwise consistency of simplex transport for categorical mediators]
\label{prop:cat:consistency:pointwise}
Under Assumptions~\ref{A1}--\ref{A4}, for any $(\boldsymbol z_0,\boldsymbol z_1)$ in the interior of the
support of $\boldsymbol Z_j$,
\[
\widehat{T}_j \bigl(p_{0,j}(\cdot \mid \boldsymbol{z}_0) \mid \boldsymbol z_0,\boldsymbol z_1 \bigr)
    \xrightarrow[]{p}
    T_j \bigl(p_{0,j}(\cdot \mid \boldsymbol{z}_0) \mid \boldsymbol z_0,\boldsymbol z_1 \bigr),
\]
where $T_j(\cdot\mid \boldsymbol z_0,\boldsymbol z_1)$ is the deterministic optimal simplex transport map
(from the conditional law under $A=0$ at $\boldsymbol z_0$ to that under $A=1$ at $\boldsymbol z_1$).
\end{prop}

\begin{cor}[Consistency along the sequential transport path]
\label{cor:cat:consistency:path}
Assume the conditions of Proposition~\ref{prop:cat:consistency:pointwise}. In the sequential procedure, set
$\boldsymbol z_1 = T_{\boldsymbol Z_j}(\boldsymbol z_0)$. Then for any $\boldsymbol z_0$ in the interior of
the support of $\boldsymbol Z_j$,
\[
\widehat{T}_j \bigl(p_{0,j}(\cdot \mid \boldsymbol{z}_0) \mid \boldsymbol z_0,\boldsymbol{z}_1 = T_{\boldsymbol{Z}_j}(\boldsymbol{z}_0) \bigr)
    \xrightarrow[]{p}
T_j \bigl(p_{0,j}(\cdot \mid \boldsymbol{z}_0) \mid \boldsymbol z_0,\boldsymbol{z}_1 = T_{\boldsymbol{Z}_j}(\boldsymbol{z}_0) \bigr).
\]
\end{cor}

\begin{prop}[Plug-in consistency of  transport along estimated sequential transport]
\label{prop:cat:consistency:plugin}
Assume~\ref{A1}--\ref{A4}, \ref{C1}--\ref{C2} and \ref{A7}. Fix $\boldsymbol z_0$ in the interior of $\mathrm{supp}(\boldsymbol Z_j)$ and let
$\boldsymbol z_1 = T_{\boldsymbol Z_j}(\boldsymbol z_0)$ and
$\widehat{\boldsymbol z}_1=\widehat T_{\boldsymbol Z_j}(\boldsymbol z_0)$.
By consistency of the sequential transport estimator for $\boldsymbol Z_j$, 
then
\[
\widehat{T}_j \bigl(p_{0,j}(\cdot \mid \boldsymbol{z}_0) \mid \boldsymbol z_0,\boldsymbol{z}_1 = \widehat{T}_{\boldsymbol{Z}_j}(\boldsymbol{z}_0) \bigr)
    \xrightarrow[]{p}
    T_j \bigl(p_{0,j}(\cdot \mid \boldsymbol{z}_0) \mid \boldsymbol z_0,\boldsymbol{z}_1 = {T}_{\boldsymbol{Z}_j}(\boldsymbol{z}_0) \bigr).
\]
\end{prop}

\subsection{Consistency of sequentially transported mediators}

Fix a topological order $\pi$ of the mediator DAG $\mathcal G$ and consider a unit $i$ in group $A=0$.
Let $\boldsymbol x^\star_{1,i}$ be the sequentially transported mediator vector defined by
\[
\begin{cases}
x^\star_{1,i,\pi(1)}=\widehat{T}_{\pi(1)}\!\left(x_{0,i,\pi(1)}\right),
\\[2mm]
x^\star_{1,i,\pi(j)}=
\widehat{T}_{\pi(j)}\!\left(
x_{0,i,\pi(j)} \,\middle|\,
\boldsymbol{x}_{0,i,\mathrm{parents}_{\mathcal{G}}(X_{\pi(j)})},\,
\boldsymbol{x}^\star_{1,i,\mathrm{parents}_{\mathcal{G}}(X_{\pi(j)})}
\right),~ j=2,\dots,d,
\end{cases}
\]
and define the population-level ST counterfactual $\boldsymbol x^\dagger_{1,i}$ by the same recursion with
$\{T_j\}_{j=1}^d$ in place of $\{\widehat T_j\}_{j=1}^d$.
In order to derive convergence results, add a final assumption \ref{A8} asking for uniform outcome regression consistency on relevant supports.

\begin{enumerate}[label=(A\arabic*)]
\setcounter{enumi}{7}
\item\label{A8}
For $a\in\{0,1\}$,
\[
\sup_{\boldsymbol x\in S_a}\bigl|\widehat\mu_a(\boldsymbol x)-\mu_a(\boldsymbol x)\bigr|\xrightarrow[]{p}0,
\]
and $\mu_a$ is continuous on $S_a$ (or at least at the relevant evaluation points).
\end{enumerate}

Assumption~\ref{A8} is a standard plug-in requirement for outcome regression: it asks that $\widehat\mu_a$ consistently estimates $\mu_a$ at the specific covariate values where it is evaluated, namely at the observed point $\boldsymbol x_{0,i}$ and at the transported point $\boldsymbol x^\dagger_{1,i}$ (whenever $\boldsymbol x^\dagger_{1,i}$ lies in the relevant support, as requested with Assumption~\ref{ass:overlap-outcome}). Such pointwise (or local) regression consistency is classical in nonparametric regression theory and underpins many plug-in arguments \citep{Stone1977}.
Compared to the double/debiased machine learning (DML) literature, Assumption~\ref{A8} is somewhat stronger: DML typically targets
average causal parameters and relies on Neyman-orthogonal scores and cross-fitting, so that nuisance functions only need to be estimated
consistently in an $L_2(\mathbb P)$ sense (often with rate conditions), rather than pointwise at particular covariate values \citep{ChernozhukovEtAl2018DML}.
We adopt the stronger formulation here to obtain a direct and transparent consistency statement for the {\em unit-level} effects $(\widehat\delta_i,\widehat\zeta_i,\widehat\tau_i)$ defined through point evaluations of $\widehat\mu_a$.

\begin{thm}[Consistency of sequential mediator counterfactuals (ST)]
\label{thm:seq:consistency}
Assume \ref{A1}--\ref{A4}. In addition, assume the following node-specific regularity conditions:
\begin{itemize}
\item for each numerical mediator $X_j$, assume~\ref{A6} and that
$F_{0,j}(x_{0,i,j}\mid \boldsymbol x_{0,i,\mathrm{parents}_{\mathcal{G}}(X_j)})\in[\varepsilon_j,1-\varepsilon_j]$ for some $\varepsilon_j\in(0,1/2)$;
\item for each categorical mediator $X_j$, assume~\ref{C1}--\ref{C2} 
\end{itemize}
Then, for the fixed unit $i$, 
\[
\boldsymbol x^\star_{1,i}\xrightarrow[]{p}\boldsymbol x^\dagger_{1,i},
\]
i.e., each coordinate $x^\star_{1,i,\pi(j)}\xrightarrow[]{p}x^\dagger_{1,i,\pi(j)}$ and hence the whole vector converges.
\end{thm}

This theorem establishes that, for any fixed individual, the transported mediator vector converges to its ideal {\em mutatis mutandis} transformation. The result is a consequence of the stability of the topological ordering and the continuity of conditional transports. The proof relies on uniform convergence of conditional transport maps and on stability of their composition along the DAG; full proofs are given in Appendix~\ref{app:proofs}. Notation $\boldsymbol{x}^{\dagger}_1$ is inspired from \cite{torous2024otdid} (where it denotes the counterfactual outcome of treated units had they not received treatment). 

Our results are stated for categorical mediators encoded as compositional data (i.e., probability vectors on the simplex), but the same consistency conclusions extend to category-valued transports as well, provided one uses a deterministic tie-breaking rule (or a mild regularization to enforce uniqueness) when mapping transported probability vectors to discrete labels.

\begin{lme}[Consistency of induced categorical assignments]
\label{cor:cat:assignment}
Assume the conditions of Proposition~\ref{prop:cat:consistency:plugin}.
Let $g:\mathcal S_K\to\{1,\dots,K\}$ be a deterministic category prediction rule (e.g., $g(\boldsymbol p)=\arg\displaystyle\max_k p_k$ with a fixed tie-breaking convention, or the deterministic cell assignment induced by the simplex transport construction).
Define the (estimated) transported probability vector, as well as the induced categorical assignments, and its population counterpart
\[
\begin{cases}
    \widehat{\boldsymbol p}^{\,\star}_{1,i,j}
:=
\widehat{T}_j \bigl(p_{0,j}(\cdot \mid \boldsymbol{z}_0) \mid \boldsymbol z_0,\boldsymbol z_1 = \widehat{T}_{\boldsymbol{Z}_j}(\boldsymbol{z}_0) \bigr),~&\widehat x^\star_{1,i,j}:=g\!\left(\widehat{\boldsymbol p}^{\,\star}_{1,i,j}\right),\\
\boldsymbol p^{\dagger}_{1,i,j}
:=
{T}_j \bigl(p_{0,j}(\cdot \mid \boldsymbol{z}_0) \mid \boldsymbol z_0,\boldsymbol z_1 = {T}_{\boldsymbol{Z}_j}(\boldsymbol{z}_0) \bigr),~&x^\dagger_{1,i,j}:=g\!\left(\boldsymbol p^{\dagger}_{1,i,j}\right).
\end{cases}
\]
Assume moreover that $g$ is locally constant at $\boldsymbol p^{\dagger}_{1,i,j}$, i.e., there exists $\gamma>0$ such that
\[
\|\boldsymbol p-\boldsymbol p^{\dagger}_{1,i,j}\|_\infty<\gamma
\quad\Longrightarrow\quad
g(\boldsymbol p)=g(\boldsymbol p^{\dagger}_{1,i,j}),
\]
then
\[
x^\star_{1,i,j}\xrightarrow[]{p}x^\dagger_{1,i,j}.
\]
\end{lme}

Note that a sufficient condition for the argmax rule is a positive margin:
$$\max_k p^{\dagger}_{1,i,j,k}-\max_{\ell\neq k} p^{\dagger}_{1,i,j,\ell} > 0.$$

\begin{remark}[Dependence on the topological order]
\label{rem:topological-order}
Sequential transport is defined by recursion along a {\em topological order} of the mediator DAG.
Any topological order is admissible, and different topological orders can only differ by permuting nodes that are {\em incomparable} in the partial order induced by $\mathcal G$ (i.e., neither is an ancestor of the other).
For such incomparable nodes $j$ and $k$, the ST update for $X_j$ does not involve $X_k$ (and vice versa), so permuting their updates leaves the population recursion unchanged; hence the population-level transported vector $\boldsymbol x^\dagger_{1,i}$ is invariant to the choice of topological order, and the estimated vector $\boldsymbol x^\star_{1,i}$ is asymptotically invariant under the consistency conditions of
Theorem~\ref{thm:seq:consistency}.
By contrast, using an order that is {\em not} topological would require transporting some mediator before one of its parents, in which case the recursion is not well-defined (or would implicitly condition on non-transported parents), and the resulting mapping would generally differ and lose the intended causal/structural interpretation.
\end{remark}

\subsection{Implications for effect decomposition}

Here, we provide convergence properties of the ST-based unit-level decomposition defined in Section~\ref{sec:st-effects-decomposition}, namely the individual indirect effects $\widehat{\delta}_i$, direct effects $\widehat{\zeta}_i$, and total effects $\widehat{\tau}_i = \widehat{\delta}_i + \widehat{\zeta}_i$.

\begin{thm}[Consistency of ST-based effect decompositions]
\label{cor:consistency:effects}
Fix a unit $i$ with $A_i=0$. Define the estimated ST-based unit-level effects by
\[
\widehat\delta_i := \widehat\mu_0(\boldsymbol x^\star_{1,i})-\widehat\mu_0(\boldsymbol x_{0,i}),\qquad
\widehat\zeta_i := \widehat\mu_1(\boldsymbol x^\star_{1,i})-\widehat\mu_0(\boldsymbol x^\star_{1,i}),\qquad
\widehat\tau_i := \widehat\delta_i+\widehat\zeta_i.
\]
Let the corresponding population-level {\em mutatis mutandis} effects be
\[
\delta_i^\dagger := \mu_0(\boldsymbol x^\dagger_{1,i})-\mu_0(\boldsymbol x_{0,i}),\qquad
\zeta_i^\dagger := \mu_1(\boldsymbol x^\dagger_{1,i})-\mu_0(\boldsymbol x^\dagger_{1,i}),\qquad
\tau_i^\dagger := \delta_i^\dagger+\zeta_i^\dagger,
\]
where $\boldsymbol x^\dagger_{1,i}$ is the population-level ST transported mediator vector.

Assume: 
(i) Assumption~\ref{ass:overlap-outcome} for some $\eta\geq 0$; 
(ii) the conditions of Theorem~\ref{thm:seq:consistency} so that $\boldsymbol x^\star_{1,i}\xrightarrow{p}\boldsymbol x^\dagger_{1,i}$; and 
(iii) Assumption~\ref{A8}. 
Then, for any $\epsilon>0$,
\[
\mathbb P\!\left(
\bigl\|(\widehat\delta_i,\widehat\zeta_i,\widehat\tau_i)-(\delta_i^\dagger,\zeta_i^\dagger,\tau_i^\dagger)\bigr\|>\epsilon
\;\middle|\; A_i=0
\right)
\le \eta + o(1).
\]
In particular, when $\eta=0$ (i.e., $\mathrm{supp}(\boldsymbol X^\dagger(1))\subseteq \mathcal{X}_0$),
\[
(\widehat\delta_i,\widehat\zeta_i,\widehat\tau_i)\xrightarrow{p}(\delta_i^\dagger,\zeta_i^\dagger,\tau_i^\dagger).
\]
\end{thm}

Thus, under mild regularity conditions, ST yields statistically sound individual-level decompositions of causal effects, without relying on a parametric structural model. The guarantees derive directly from the stability and convergence of the underlying conditional transport maps.

The next sections provide empirical illustrations in Gaussian settings and evaluate the method on real-world data under a fairness-motivated causal graph.

\section{Gaussian illustration and simulation studies}
\label{sec:6:gaussian-sims}

We now investigate the behaviour of ST in controlled settings.
We begin with Gaussian examples where the mediation decomposition is analytically tractable, allowing direct comparison between ST-based counterfactuals and classical linear SEMs. We then present Monte Carlo experiments examining the statistical performance of ST under varying sample sizes, mediator dimensions, and DAG structures.

\subsection{Gaussian mediation with a single mediator}
\label{sec:single:gaussian}

Consider the simple mediation model
\[
A \to X \to Y,
\qquad
A \to Y,
\]
with
\[
X = \alpha A + \varepsilon_X,
\qquad
Y = \beta X + \gamma A + \varepsilon_Y,
\]
where $(\varepsilon_X,\varepsilon_Y)$ are jointly Gaussian with zero mean and covariance matrix independent of~$A$. Under this model, the natural direct and indirect effects coincide with the algebraic SEM decomposition:
\[
\delta=\text{NIE} = \alpha \beta,
\qquad
\zeta=\text{NDE} = \gamma.
\]

Consider a unit $i$ with $a_i = 0$, and denote the realization of $X_i$ by $x_i = x_{0,i}$. Applying ST to this simple mediation model, we can compute the ST-based counterfactual value of unit $i$ under $A = 1$, denoted $x_i^\star = x^\star_{1,i}$, as
\[
x_{1,i}^\star = T(x_{0,i}) = Q_{1}\!\circ F_{0}(x_{0,i}),
\]
which corresponds to the classical univariate OT map.

$T$ coincides with the affine transformation that maps the $A=0$ distribution of $X$ to that of $A=1$, recovering the SEM counterfactual calculation $x_i(1) = \alpha + x_i(0)$. 

Given a regression estimator $\widehat{\mu}_a(x)$ consistent for $\mathbb{E}[Y \mid A=a,X=x]$, the ST-based decomposition
\[
\delta_i = \widehat{\mu}_0(x^\star_{1,i}) - \widehat{\mu}_0(x_{0,i}),
\qquad
\zeta_i = \widehat{\mu}_1(x^\star_{1,i}) - \widehat{\mu}_0(x^\star_{1,i}),
\]
satisfies
\[
\mathbb{E}[\delta_i] \approx \alpha \beta,
\qquad
\mathbb{E}[\zeta_i] \approx \gamma.
\]
Hence, ST recovers classical mediation effects in the linear-univariate-Gaussian case.
The agreement holds exactly in population and approximately in finite samples, with discrepancies attributable only to estimation error in the empirical transport map.

\subsection{Multivariate gaussian mediators with a DAG structure}

We now consider $d$ mediators $\boldsymbol{X} = (X_1,\dots,X_d)$ arranged in a DAG $\mathcal{G}$. Let
\[
\boldsymbol{X} = B \boldsymbol{X} + c A + \varepsilon,
\qquad
Y = \theta^\top \boldsymbol{X} + \gamma A + \eta,
\]
where $B$ is strictly upper-triangular in accordance with the DAG ordering, $c\in\mathbb{R}^d$, and all disturbances are Gaussian. In this setting, the mediator distribution under each treatment is multivariate Gaussian with distinct means but a common covariance matrix determined by $(B,\mathrm{Var}(\varepsilon))$.

Because ST applies univariate or conditional OT maps along the DAG, each transported coordinate takes the form
\[
x^\star_{1,j}
    = \mu_{1,j}(\boldsymbol{z}) + 
      \frac{\sigma_{1,j}(\boldsymbol{z})}{\sigma_{0,j}(\boldsymbol{z})}
      \bigl( x_{0,j} - \mu_{0,j}(\boldsymbol{z}) \bigr),
\]
where $\boldsymbol{z}=\mathrm{parents}_{\mathcal{G}}(x_j)$ and $(\mu_{a,j}(\boldsymbol{z}),\sigma_{a,j}(\boldsymbol{z}))$ are the Gaussian conditional mean and standard deviation. Thus ST reproduces the affine transformation induced by the structural equations, yielding perfect agreement (up to estimation error) with the SEM-based mediation decomposition.

Numerical illustrations for the two-dimensional case ($d=2$) are presented in Appendix~\ref{app:toy}.

\subsection{Monte Carlo simulations with three mediators}
\label{sec:sims}

We assess the finite-sample performance of ST in a mixed-type mediation setting with nonlinear outcome models and causally dependent mediators. The data-generating process is designed to reflect the structure studied in Section~\ref{sec:general:DAG}, with two numerical mediators, one categorical mediator, and a binary outcome.

\paragraph*{Data generating process.}
Each Monte Carlo replication, performed using the \texttt{R} software (version 4.5.2) consists of simulating i.i.d.\ observations of
$$
(A,\boldsymbol{X},Y)=
(A, X_1, X_2, X_3, Y),
$$
where $A\in\{0,1\}$ is a binary treatment, $(X_1,X_2)\in\mathbb{R}^2$ are continuous mediators, $X_3\in\{\mathrm{A},\mathrm{B},\mathrm{C}\}$ is a categorical mediator, and $Y\in\{0,1\}$ is a binary outcome.

Conditional on $A=a$, the continuous mediators $(X_1,X_2)$ follow a bivariate normal distribution with treatment-specific mean vector $\mu_a$ and covariance matrix $\Sigma_a$, inducing both location and dependence shifts across treatment groups. The categorical mediator $X_3$ is generated conditionally on $(X_1,X_2,A)$ via a multinomial logistic model, so that $X_3$ is causally downstream of both numerical mediators and of the treatment. The outcome $Y$ is then drawn from a Bernoulli distribution with success probability given by a logistic regression whose coefficients depend on the treatment level and include nonlinear and non-additive effects of the mediators. Full details of the DGP are reported in Appendix~\ref{app:toy-example-three-mediators}.

The resulting causal structure corresponds to the DAG shown in Figure~\ref{fig:dag-mediation-example}, with
$$
A \;\to\; (X_1,X_2,X_3) \;\to\; Y,
\qquad
(X_1,X_2) \;\to\; X_3,
$$
and direct effects of $A$ on $Y$. In each replication, we generate $n_0=400$ observations with $A=0$ and $n_1=200$ observations with $A=1$.

\begin{figure}[htb!]
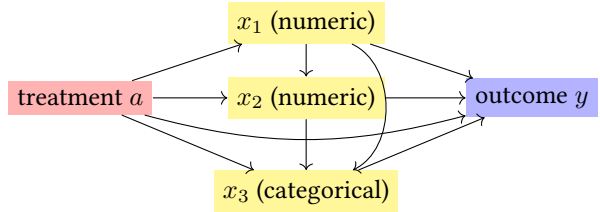

    \centering
    \tikzset{node distance=1.5cm and 2.5cm}
    \tikz{
        \node[fill=red!30]   (a)  at (0,1.5) {treatment $a$};
        \node[fill=yellow!50](x1) at (3,2.5)   {$x_1$ (numeric)};
        \node[fill=yellow!50](x2) at (3,1.5) {$x_2$ (numeric)};
        \node[fill=yellow!50](x3) at (3,.25)   {$x_3$ (categorical)};
        \node[fill=blue!30]  (y)  at (6,1.5) {outcome $y$};

        \path[->, black] (a) edge (x1);
        \path[->, black] (a) edge (x2);
        \path[->, black] (a) edge (x3);
        \path[->, black, bend right=15] (a) edge (y);

        \path[->, black] (x1) edge (x2);
        \path[->, black, bend left=65] (x1) edge (x3);
        \path[->, black] (x2) edge (x3);


        \path[->, black] (x1) edge (y);
        \path[->, black] (x2) edge (y);
        \path[->, black] (x3) edge (y);
    }
    \caption{DAG for the simulated data.}
    \label{fig:dag-mediation-example}
\end{figure}

\paragraph*{Counterfactual construction and benchmarks.}
For each simulated dataset, we construct mediator counterfactuals for individuals
in the untreated group ($A=0$) using three transport-based approaches that yield individual-level counterfactual mediators: multivariate optimal transport without conditioning, its entropic-regularized variant, and sequential transport. Sequential transport is applied along the topological order
$X_1 \rightarrow X_2 \rightarrow X_3$ induced by the mediator DAG.
For ST, numerical mediators are transported using (conditional) quantile mapping,
while the categorical mediator is handled via simplex-based deterministic transport
as described in Section~\ref{sec:general:DAG}.

As a reference, we additionally report classical causal mediation analysis based on linear structural equation models, implemented via the
\texttt{mediation} package \citep{JSSv059i05},  as well as the quantile-preservation approach using quantile regression proposed by \citet{plevcko2020fair} in the \texttt{fairadapt} package \citep{plevcko2021fairadapt}. For the latter, the average causal mediation effect is estimated directly and does not rely on constructing individual-level mediator counterfactuals.

\paragraph*{Outcome regression and effect estimation.}
For each replication, we fit two outcome models by random forests, separately by treatment group (assuming Assumption~\ref{ass:overlap-outcome}):
$$
\widehat\mu_0(\boldsymbol{x}) \approx \mathbb{E}[Y\mid A=0,\boldsymbol{X}=\boldsymbol{x}],
\qquad
\widehat\mu_1(\boldsymbol{x}) \approx \mathbb{E}[Y\mid A=1,\boldsymbol{X}=\boldsymbol{x}].
$$
Given a set of counterfactual mediators for untreated units, denoted $\boldsymbol{x}_{1,i}^\star$ for individuals with $A_i=0$, we compute the ST-based decomposition on the untreated as defined in Section~\ref{sec:st-effects-decomposition}. In particular, we evaluate the indirect effect $\delta_{i}$, the direct effect $\zeta_{i}$,
and the total effect $\tau_i=\delta_{i}+\zeta_{i}$ using the fitted outcome models $\widehat\mu_0$ and $\widehat\mu_1$. This plug-in evaluation appeals to Assumption~\ref{ass:overlap-outcome}, ensuring that $\hat\mu_0(\cdot)$ is supported on the transported region.
We report the corresponding averages
$$
\bar\delta=\frac{1}{n_0}\sum_{i:A_i=0}\delta_{i},\quad
\bar\zeta=\frac{1}{n_0}\sum_{i:A_i=0}\zeta_{i},\quad
\bar\tau=\bar\delta+\bar\zeta.
$$
When counterfactual mediators for treated units are also available, we analogously compute $\bar\delta$ and $\bar\zeta$ by swapping the roles of groups~0 and~1. As a robustness diagnostic, we additionally report versions of $\bar\delta$ obtained by replacing $\widehat\mu_a(\boldsymbol{x}_{a,i})$ with the observed outcome $Y_i$, which allows us to assess sensitivity to outcome-model estimation. 

\paragraph*{Monte Carlo protocol and outputs.}
The entire procedure (data generation, counterfactual construction, outcome-model fitting, and effect estimation) is repeated over $B=200$ independent Monte Carlo replications. For each method, we report the empirical distribution of $(\bar{\delta},\bar{\zeta},\bar{\tau})$ across replications, as well as the dispersion of individual-level effects for approaches that yield deterministic unit-level decompositions.

The objective of these simulations is to evaluate the stability, interpretability, and finite-sample behavior of sequential transport in a setting with mixed mediator types, nonlinear outcome models, and a nontrivial mediator DAG. Results are summarized in Figure~\ref{fig:xp-simul:delta:zeta:tau}.

\paragraph*{Evaluation criteria.}
The Monte Carlo study focuses on the empirical behavior of the estimated effect
decompositions across replications. For each method, we examine:
\begin{itemize}
    \item the Monte Carlo distribution of the averaged effects
    $\bar\delta$, $\bar\zeta$, and $\bar\tau$;
    \item the dispersion and skewness of these quantities across replications;
    \item the stability of individual-level effects $\delta_{0,i}$ and
    $\zeta_{1,i}$ for methods that produce deterministic unit-level counterfactuals;
    \item qualitative agreement with linear SEM-based mediation in settings where
    the data-generating process is approximately linear-Gaussian.
\end{itemize}

\paragraph*{Results.}
Figure~\ref{fig:xp-simul:delta:zeta:tau} displays the empirical distributions of
$\bar\delta$, $\bar\zeta$, and $\bar\tau$ across the $200$ Monte Carlo replications.
In this mixed-Gaussian simulation setting, all transport-based approaches yield broadly similar
distributions. Sequential transport (ST) shows a modest reduction in dispersion compared to
multivariate optimal transport (OT) and entropic optimal transport (SKH), an effect that is more clearly seen in numerical
summaries (e.g., SD/IQR) than from the violin plots alone.

When compared to classical causal mediation analysis based on linear structural equation models (CM), ST yields average effects of similar magnitude in this approximately linear-Gaussian setting, despite relying on a fully nonparametric, distributional construction of mediator counterfactuals. Differences between ST and CM become more pronounced for individual-level effects, which are only available for transport-based methods.

In the presence of mixed mediator types, including a categorical mediator handled via simplex transport, ST remains numerically stable and produces coherent counterfactual mediator values. While variability increases relative to purely continuous settings, this increase is comparable to that observed for other transport-based methods and does not lead to pathological behavior. Overall, these results suggest that sequential transport provides a stable and interpretable finite-sample decomposition of effects in mediation settings with heterogeneous mediators and nontrivial causal structure.

\begin{figure}[htb!]
    \centering
    \includegraphics[width=\linewidth]{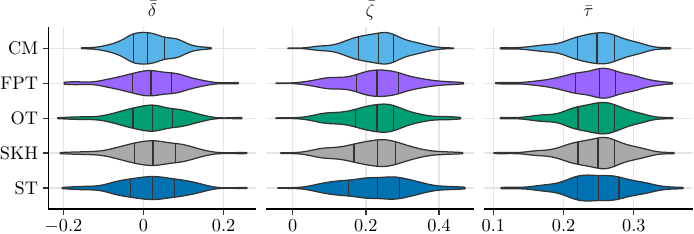}
    \caption{Estimation (200 replicated samples, $n=600$) of average indirect effect $\bar{\delta}$, average direct effect $\bar{\zeta}$ and average total causal effect $\bar{\tau}$.}
    \label{fig:xp-simul:delta:zeta:tau}
\end{figure}

\subsection{Discussion}

These experiments confirm that ST reproduces classical mediation decompositions in Gaussian settings while extending naturally beyond the linear regime and to mixed mediator types. The results provide empirical support for the theoretical guarantees of Section~\ref{sec:theory}, demonstrating that ST-based counterfactuals converge to their target {\em mutatis mutandis} transformations and yield stable effect decompositions in low dimension but with different mediator types. 

The next section applies ST to a real-world fairness problem, where mediators include both ordinal and categorical variables structured by a causal DAG.

\section{Application to the COMPAS recidivism dataset}
\label{sec:7:realdata}

We illustrate the sequential transport framework on a well-known benchmark in the algorithmic fairness literature, the COMPAS dataset \citep{Larson2016}. This dataset contains
individual-level information used to predict whether a criminal defendant is likely to reoffend within two years (recidivism within two years after parole start). Our goal is not to revisit the debate surrounding COMPAS risk scores, but to demonstrate how ST yields distributionally-grounded counterfactual mediators and produces individual-level decompositions of group disparities along the causal pathways specified by an interpretable DAG.

\subsection{Data and causal graph}

We consider the standard pretrial subset originally analyzed in \cite{Kusner17}, provided with the \texttt{R} package \texttt{fairadapt} \citep{plevcko2021fairadapt}, following the preprocessing conventions used in recent transport-based fairness studies \citep{charpentier2023optimal,
machado2025categorical,machado2025sequential}. The outcome $Y$ is the two-year general recidivism indicator, which takes the value $1$ if a criminal defendant re-offended within two years after parole start, and $0$ otherwise. The sensitive attribute is
\[
A = \texttt{race} \in \{0,1\},
\]
where $a=1$ denotes white defendants and $a=0$ non-white defendants. In this application, we treat $A$ as a root node by convention (we do not model its upstream causes), but we do not interpret the resulting decomposition as an identified causal effect of race; rather, it is a distributional pathway decomposition aligned with the mediator DAG.

The mediator set includes both numerical and categorical variables:
\begin{align*}
  \boldsymbol{X} = 
\{&\texttt{age},\,
  \texttt{juv\_fel\_count},\,
  \texttt{juv\_misd\_count},\,
  \texttt{juv\_other\_count},\,\\
&  \texttt{priors\_count},\,
  \texttt{charge\_degree},\,
  \texttt{charge\_degree\_bins} \}.  
\end{align*}
Following \cite{plevcko2021fairadapt},we use a simplified mediator set $$  \boldsymbol{X}=\{\texttt{age},\ \texttt{priors\_count},\ \texttt{charge\_degree}\}$$ and the corresponding mediator DAG shown in Figure~\ref{fig:dag-compas-example}, in which race influences \texttt{age}, \texttt{priors\_count}, and \texttt{charge\_degree}, while the remaining mediators follow the dependency structure encoded by a sparse DAG consistent with judicial and criminal-history processes. This graph offers a reasonable compromise between causal interpretability and empirical realism.


\begin{figure}[t!]
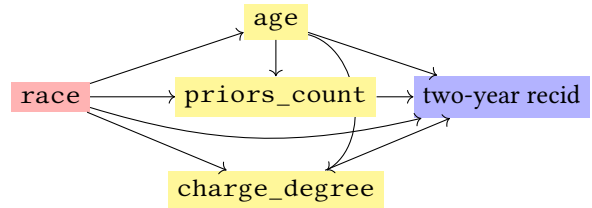

    \centering
    \tikzset{node distance=1.5cm and 2.5cm}
    \tikz{
        \node[fill=red!30]   (a)  at (0,1.5) {\texttt{race}};
        \node[fill=yellow!50](x1) at (3,2.5)   {\texttt{age}};
        \node[fill=yellow!50](x2) at (3,1.5) {\texttt{priors\_count}};
        \node[fill=yellow!50](x3) at (3,.25)   {\texttt{charge\_degree}};
        \node[fill=blue!30]  (y)  at (6,1.5) {two-year recid};

        \path[->, black] (a) edge (x1);
        \path[->, black] (a) edge (x2);
        \path[->, black] (a) edge (x3);
        \path[->, black, bend right=15] (a) edge (y);

        \path[->, black] (x1) edge (x2);
        \path[->, black, bend left=68] (x1) edge (x3);

        \path[->, black] (x1) edge (y);
        \path[->, black] (x2) edge (y);
        \path[->, black] (x3) edge (y);
        \node[fill=yellow!50](x2) at (3,1.5) {\texttt{priors\_count}};
    }
    \caption{Simplified DAG for the \texttt{COMPAS} dataset.}
    \label{fig:dag-compas-example}
\end{figure}

\subsection{Sequential transport in practice}

Using Algorithm~\ref{alg:0}, we estimate conditional transport maps for numerical mediators via kernel-based optimal transport, and apply simplex-based transport for categorical mediators \citep{machado2025categorical}. The maps are applied in a topological order of the DAG, yielding for each individual $i$ in the non-white group a transported counterfactual vector $\boldsymbol{x}_{1,i}^\star$ representing how their mediator profile would need to change to be aligned with the distribution of mediators among white defendants, {\em {\em mutatis mutandis}}.

For the outcome regression, we fit two random forest models $\widehat{\mu}_0(\boldsymbol{x})$ and $\widehat{\mu}_1(\boldsymbol{x})$, each trained on observations from the corresponding treatment group. These flexible regressors capture potential nonlinear effects of mediators on recidivism without imposing a parametric structure.

\subsection{Decomposition of group disparities}


For individuals with $A_i=0$, we compute the ST-based unit-level decomposition defined in Section~\ref{sec:st-effects-decomposition}, using the fitted outcome models $\widehat{\mu}_0$ and $\widehat{\mu}_1$ and the transported mediator values $\boldsymbol{x}_{1,i}^\star$.
This yields individual indirect effects $\delta_i$, direct effects $\zeta_i$, and total effects $\tau_i=\delta_i+\zeta_i$.

Under Assumption~3.1, the evaluation of $\widehat{\mu}_0$ at
$\boldsymbol{x}_{1,i}^\star$ remains within regions supported by $A=0$ data, which mitigates extrapolation.

Aggregating over the untreated group produces
$\bar{\delta}$ (ST-based indirect disparity),
$\bar{\zeta}$ (ST-based direct disparity),
and the total disparity $\widehat{\tau}=\bar{\delta}+\bar{\zeta}$.

Figure~\ref{fig:decomp:compas} shows the estimated individual effects, Table~\ref{tab:average-causal-effects-compacts} reports the aggregated effects. ST attributes approximately two-thirds of the disparity $\widehat{\tau}$ to differences in mediator distributions. This is consistent with previous findings \citep{Kusner17, charpentier2023optimal}, while providing a richer breakdown at the individual and mediator levels.

\begin{figure}[t!]
    \centering
    \includegraphics[width=\linewidth]{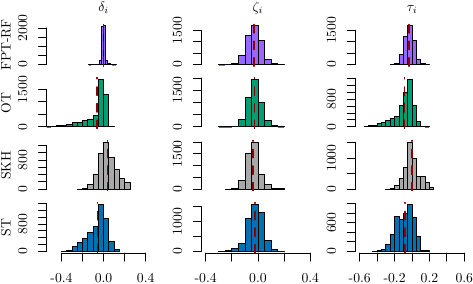}
    \caption{Individual {\em mutatis mutandis} effects (distributional decomposition) for the \texttt{COMPAS} dataset estimated with counterfactual-based methods using Fairadapt (FPT-RF), optimal transport (OT), penalized OT (SKH), and sequential transport (ST) according to the causal graph of Figure~\ref{fig:dag-compas-example}.}
    \label{fig:decomp:compas}
\end{figure}

\begin{table}[t!]
    \centering
    \begin{tabular}[t]{rrrrrrr}
\toprule
Metric & CM & FPT & NF & OT &  SKH & ST\\
\midrule
$\bar{\delta}$ & -0.09 & 0.00 & {-0.06} & -0.06 &  0.04 & -0.06\\
$\bar{\zeta}$  & 0.00  &-0.03 & 0.00 & -0.03 & -0.04 & -0.02\\
$\bar{\tau}$   & -0.09 &-0.03 & -0.06 & -0.09 &  0.00 & -0.08\\
\bottomrule
\end{tabular}
    \caption{Estimated average causal effects for \texttt{COMPAS}, based on causal mediation analysis with LSEM using \texttt{mediation} package (CM), with counterfactual-transport-based-methods: optimal transport (OT), penalized OT (SKH), sequential transport (ST) and fairadapt (FPT) according to the DAG of Figure~\ref{fig:dag-compas-example}, and using normalizing flows (NF).}
    \label{tab:average-causal-effects-compacts}
\end{table}

\subsection{Mediator-level attribution}

Since ST transports mediators sequentially, we may examine the marginal contribution of each mediator $x_j$ by evaluating the effect of transporting mediators only up to that node:
\[
\delta_i^{(j)} 
  = \widehat{\mu}_0\!\left(\boldsymbol{x}_{1,i}^{\star,(j)}\right)
    - \widehat{\mu}_0(\boldsymbol{x}_{0,i}),
\]
where $\boldsymbol{x}_{1,i}^{\star,(j)}$ denotes the partially transported vector obtained by applying the maps for $x_{\pi(1)},\dots,x_{\pi(j)}$ but leaving all later mediators at their factual values. The same overlap requirement applies to partially transported profiles; we again rely on Assumption~\ref{ass:overlap-outcome}.
This yields a deterministic and DAG-coherent decomposition of the indirect effect across mediators.

Empirically, the largest contributions arise from \texttt{priors\_count} and \texttt{charge\_degree}, both of which lie close to the root of the DAG and exhibit substantial distributional shifts across racial groups. Variable \texttt{age} contributes modestly, reflecting age differences between groups and its nonlinear effect on recidivism risk.

\subsection{Discussion}

As can be seen from the results, sequential transports yields a structurally coherent and distributionally grounded decomposition of group disparities. The method provides:
\begin{itemize}
    \item deterministic counterfactual mediator profiles consistent with the causal DAG;
    \item unit-level estimates of direct and indirect effects;
    \item a natural attribution of indirect effects across individual mediators;
    \item compatibility with complex mediator types (categorical, ordinal, numerical);
    \item robustness to nonlinear outcome models.
\end{itemize}

The approach remains subject to the usual identification challenges of observational studies. In the presence of unobserved confounding between race and the mediators, the estimated decompositions should be interpreted as
distributional rather than causal. Nevertheless, ST provides a principled construction of {\em mutatis mutandis} counterfactuals that remain in-distribution and respect the assumed mediator DAG. As such, sequential transport yields an interpretable decomposition of observed disparities along the mediator pathways encoded in the DAG.

\section{Discussion and extensions}
\label{sec:discussion}

This paper introduced a distributional framework for causal mediation analysis based on OT and its sequential extension along a mediator DAG. By replacing cross-world counterfactuals with {\em mutatis mutandis} transport-based mediator profiles, the method yields individual-level decompositions of treatment effects while avoiding the  need for a fully specified structural causal model with functional forms. Sequential transport produces counterfactual mediators that remain in-distribution and respect the conditional structure encoded in the DAG between the mediators, and accommodates both numerical and categorical mediators. Theoretical guarantees establish the consistency of the induced maps and the associated effect decompositions, and empirical results confirm similar performance in both Gaussian and non-Gaussian regimes, as well as in a real-world fairness application, compared to commonly used structural counterfactuals.

\subsection{Interpretation and limitations}

The framework is primarily {\em distributional}: it constructs mediator counterfactuals by interpolating between empirical conditional distributions rather than by assuming a structural mechanism. As such, ST inherits both the strengths and the limitations of distributional counterfactuals. When the treatment is exogenous, as in the mediation setting considered here, the {\em mutatis mutandis} interpretation aligns naturally with the graphical structure and yields interpretable decompositions. In observational settings with unobserved confounding between treatment and mediators, the counterfactuals remain well defined but do not necessarily correspond to identified causal effects. This is a feature shared with transport-based fairness methods \citep{charpentier2023optimal,machado2025sequential}.

A second limitation is the reliance on accurate estimation of conditional distributions.
Although kernel estimators and simplex-based transports perform well in moderate dimensions, high-dimensional conditioning sets may lead to slower convergence or greater instability. This is mitigated in practice by the sparsity of mediator DAGs, but remains a potential challenge for fully connected graphs or large parent sets.

Finally, while ST produces deterministic counterfactuals, the method is sensitive to the choice of DAG structure. Misspecification of mediator relationships may distort the conditional transports. In many applications, however, causal graphs are either known from domain knowledge or can be approximated by sparse structures capturing essential dependencies \citep{bulhmann2007, Toth_2022_neurips, cai2023on}.

\subsection{Extensions to longitudinal mediators}

Although our analysis focuses on static treatment and mediator variables, the framework extends naturally to time-varying settings. In longitudinal mediation, mediators, treatments, and confounders may evolve over time, and causal estimands such as path-specific or interventional effects become more complex. The semiparametric framework of \citet{diaz2023efficient} provides identification and estimation strategies for such settings. Combining their approach with transport-based methodology suggests that one may construct {\em dynamic {\em mutatis mutandis} counterfactuals} by transporting mediator trajectories sequentially across time.

A straightforward extension would posit a family of time-indexed maps
\[
T_t: (X_{0,t}, \boldsymbol{Z}_{0,t}) \mapsto (X_{1,t}, \boldsymbol{Z}_{1,t}),
\]
where $\boldsymbol{Z}_{t}$ denotes parents or past mediators, and apply these maps recursively along a longitudinal DAG. More ambitious developments involve continuous-time or stochastic interpolations between treatment regimes, using dynamic OT or Schr\"odinger bridge formulations \citep{benamou2000,leonard2014}. Such methods construct controlled stochastic processes whose trajectories interpolate between mediator distributions under different treatment histories, offering a principled analogue of dynamic structural counterfactuals \citep{robins1986,robins2000,murphy2003}. A rigorous integration of these ideas with longitudinal mediation estimands remains an open and promising direction.

\subsection{Beyond mediator DAGs}

ST naturally extends to richer causal structures. Possible extensions
include:
\begin{itemize}
    \item \textbf{Multiple treatments or multivalued interventions}, where maps are estimated for each pair of treatment levels and composed along treatment paths \citep{imai2010general}.
    \item \textbf{Interventional, stochastic, or policy-based effects}, where the target mediator distribution is not empirical but corresponds to a policy or intervention, enabling counterfactual predictions under hypothetical reforms \citep{didelez2006,vansteelandt2017}.
    \item \textbf{Path-specific effects}, where only a subset of mediator pathways is transported. This aligns ST with the framework of path-specific fairness constraints and mediated causal pathways explored in the fairness literature \citep{avin2005identifiability,shpitser2013counterfactual}.
    \item \textbf{Hybrid generative–transport models}, combining structural equations for some mediators with transport maps for others, useful when partial domain knowledge is available \citep{moselhy2012,parno2018}.
\end{itemize}

\subsection{Concluding remarks}

Sequential transport offers a flexible, interpretable, and distributionally robust approach to causal mediation analysis. By grounding counterfactual mediators in optimal transport and in the structure of a mediator DAG, the framework bridges classical mediation analysis, modern transport methods, and emerging applications such as algorithmic fairness. Extensions to longitudinal settings, dynamic interventions, and richer causal structures promise to deepen this connection, making transport-based counterfactuals a broadly applicable tool in causal inference.


\bibliographystyle{apalike}
\bibliography{biblio}
\appendix

\section{Additional definitions and mathematical background}
\label{app:math}

This section provides the mathematical foundations underlying the sequential transport (ST) framework: conditional optimal transport (OT), conditional quantile operators, simplex-based transport for categorical mediators, and formal definitions needed in the proofs.

\subsection{Conditional distributions and quantile operators}

Let $(X,\boldsymbol{Z},A)$ denote a mediator $X$, its parent set $\boldsymbol{Z}=\mathrm{parents}_{\mathcal{G}}(X)$, and the treatment variable $A\in\{0,1\}$. For $a\in\{0,1\}$, denote by
\[
F_{a}(x \mid \boldsymbol{z}) = \mathbb{P}(X \le x \mid \boldsymbol{Z}=\boldsymbol{z}, A=a)
\]
the conditional cumulative distribution function (CDF), assumed continuous and strictly increasing in $x$ on the support of $X$.

The conditional quantile operator is
\[
Q_{a}(u \mid \boldsymbol{z}) = \inf\{ x : F_{a}(x \mid \boldsymbol{z}) \ge u \},\qquad u\in(0,1).
\]
Under monotone transport, the unique optimal map pushing $X\mid A=0$ to $X\mid A=1$ given $\boldsymbol{Z}=\boldsymbol{z}$ is
\[
T(x\mid \boldsymbol{z}) = Q_{1}(F_{0}(x\mid \boldsymbol{z})\mid \boldsymbol{z}).
\]

\subsection{Kernel estimation of conditional CDFs and quantiles}

Let $\{(x_i,\boldsymbol{z}_i)\}_{i=1}^n$ denote samples under $A=a \in \{0,1\}$. The conditional CDF is estimated by
\[
\widehat{F}_{a}(x\mid \boldsymbol{z})
    = \frac{\sum_{i=1}^n K_h(\boldsymbol{z}-\boldsymbol{z}_i)\,
            \mathbb{I}\{x_i \le x\}}
           {\sum_{i=1}^n K_h(\boldsymbol{z}-\boldsymbol{z}_i)},
\]
with $K_h(\cdot)$ a smoothing kernel. Conditional quantiles under $A=a$ are estimated by solving
\[
\widehat{Q}_{a}(u\mid \boldsymbol{z})
    = \inf\{ x : \widehat{F}_{a}(x\mid \boldsymbol{z}) \ge u \}.
\]
Under standard assumptions (bounded kernels, $h\to 0$, $nh\to\infty$), the estimators $\widehat{F}_0$ and $\widehat{Q}_{1}$  converge uniformly.

\subsection{Continuous transport}
\label{app:OT:continuous}

Let $\nu_0$ and $\nu_1$ be probability measures on $\mathbb{R}^d$ with finite second moments. Given a measurable cost function $c:\mathbb{R}^d\times\mathbb{R}^d\to\mathbb{R}$, the Kantorovich OT problem seeks a coupling $\pi$ between $\nu_0$ and $\nu_1$ solving
\begin{equation}
\label{eq:kantorovich}
\inf_{\pi\in\Pi(\nu_0,\nu_1)}
\int_{\mathbb{R}^d\times\mathbb{R}^d}
c(x,y)\, \mathrm d\pi(x,y),
\end{equation}
where $\Pi(\nu_0,\nu_1)$ denotes the set of all joint probability measures on $\mathbb{R}^d\times\mathbb{R}^d$ with marginals $\nu_0$ and $\nu_1$. When an OT map $T$ exists, i.e.\ when there is a measurable function $T:\mathbb{R}^d\to\mathbb{R}^d$ such that $T_\#\nu_0=\nu_1$, the Monge formulation
\begin{equation}
\label{eq:monge}
\inf_{T:\,T_\#\nu_0=\nu_1}
\int_{\mathbb{R}^d}
c\bigl(x,T(x)\bigr)\, \mathrm d\nu_0(x)
\end{equation}
is equivalent to \eqref{eq:kantorovich}. In one dimension, and for convex costs such as the squared Euclidean cost, the optimal map is uniquely defined almost everywhere and admits a closed-form expression in terms of distribution and quantile functions.

In the univariate case, let $\nu_0$ and $\nu_1$ be Borel probability measures on $\mathbb{R}$ with continuous and strictly increasing CDFs $F_0$ and $F_1$.
Under squared Euclidean cost $c(x,y)=\|x-y\|^2$, the (unique a.e.) optimal Monge map $T:\mathbb{R}\to\mathbb{R}$ pushing $\nu_0$ to $\nu_1$ satisfies
\[
T_\# \nu_0 = \nu_1,
\qquad
T(x) = Q_1(F_0(x)),
\]
where $Q_1$ denotes the quantile function (generalized inverse) of $F_1$.
This monotone rearrangement is the population counterpart of the univariate transport used in our mediator counterfactual construction \cite{villani2003optimal, peyre2019computational}.

If $\nu_0=\mathcal N(m_0,\sigma_0^2)$ and $\nu_1=\mathcal N(m_1,\sigma_1^2)$ then
\[
T(x)=m_1+\frac{\sigma_1}{\sigma_0}(x-m_0),
\]
recovering the familiar affine transformation in the Gaussian case.

\subsection{Discrete transport and the Birkhoff polytope}
\label{app:OT:discrete}

Discrete OT provides a complementary view that is useful both computationally and
conceptually. Let $\{x_{0,i}\}_{i=1}^{n_0}$ and $\{x_{1,j}\}_{j=1}^{n_1}$ be empirical
samples with weights $w_0\in\mathbb R_+^{n_0}$ and $w_1\in\mathbb R_+^{n_1}$,
$\sum_i w_{0,i}=1$, $\sum_j w_{1,j}=1$. Define the cost matrix
$C\in\mathbb R^{n_0\times n_1}$ with entries $C_{ij}=c(x_{0,i},x_{1,j})$.
A transport plan $P$ is a nonnegative matrix with prescribed marginals:
\[
\mathcal U(w_0,w_1)
=
\Bigl\{P\in\mathbb R_+^{n_0\times n_1}:\;
P\mathbf 1_{n_1}=w_0,\;
P^\top \mathbf 1_{n_0}=w_1
\Bigr\}.
\]
The discrete Kantorovich problem is
\begin{equation}
\label{eq:discrete-ot}
P^\star \in \arg\min_{P\in\mathcal U(w_0,w_1)}\lbrace \langle P,C\rangle\rbrace,
\text{ where }
\bigl\langle P,C\bigr\rangle=\sum_{i,j} P_{ij}C_{ij}.
\end{equation}
When $w_0=\frac1{n_0}\mathbf 1$ and $w_1=\frac1{n_1}\mathbf 1$, $\mathcal U(w_0,w_1)$
reduces to a (rectangular) Birkhoff polytope. A plan can be turned into a proxy map
through the barycentric projection
\begin{equation}
\label{eq:barycentric-projection}
T_P(x_{0,i})=\frac{1}{w_{0,i}}\sum_{j=1}^{n_1} P_{ij}x_{1,j},
\end{equation}
which coincides with a hard matching when $P^\star$ is (approximately) sparse.

\subsection{Penalized and weighted optimal transport}\label{app:pen:weig}

For $\gamma>0$ the {\em entropic optimal transport problem} between two measures $\nu_0$ and $\nu_1$ is
\[
\pi^{(\gamma)}
  \;=\;
  \arg\min_{\pi\in\Pi(\nu_0,\nu_1)}
  \bigl\{\,\int{c}\mathrm{d}{\pi} 
          + \gamma\,\mathrm{KL}\!\bigl(\pi \;\big\|\; \nu_0\otimes\nu_1\bigr)
  \bigr\},
\]
where $\Pi(\nu_0,\nu_1)$ is the set of couplings with the prescribed marginals.

\paragraph*{Penalized Gaussian, univariate case}

Write
$
\nu_0=\mathcal{N}(m_0,\sigma_0^{2}),\;
\nu_1=\mathcal{N}(m_1,\sigma_1^{2})
$.
Then $\pi^{(\gamma)}$ is the bivariate normal
\[
\pi^{(\gamma)}
  =\mathcal{N}\!\Bigl(
      \begin{pmatrix}m_0\\m_1\end{pmatrix},
      \begin{pmatrix}
        \sigma_0^{2} & \rho^{(\gamma)}\sigma_0\sigma_1\\
        \rho^{(\gamma)}\sigma_0\sigma_1 & \sigma_1^{2}
      \end{pmatrix}
    \Bigr),
\]
with correlation

\[
\rho^{(\gamma)}
=\frac{1}{\sigma_0\sigma_1} \sqrt{(\frac{\gamma}{4})^2+\sigma_0^2\sigma_1^2}-\frac{1}{\sigma_0\sigma_1}\frac{\gamma}{4}.
\]
Hence, if $Y_0\sim\nu_0$ and $Y_1\sim\nu_1$,
$\mathrm{Cov}(Y_0,Y_1)=\rho^{(\gamma)}\sigma_0\sigma_1$.
When $\gamma\to0$, $\rho^{(\gamma)}\to 1$ and the plan collapses
onto the classical (comonotone) OT map $T(y)=m_1+\displaystyle\frac{\sigma_1}{\sigma_0}(y-m_0)$.

\paragraph*{Penalized Gaussian, multivariate case}

In the multivariate setting, the optimal coupling is a {Gaussian measure}
\[
\pi^{(\gamma)}
=\mathcal{N}\!\Bigl(
      \begin{pmatrix}m_0\\m_1\end{pmatrix},
      \begin{pmatrix}
        \Sigma_0 & \Gamma^{(\gamma)}\\[2pt]
        (\Gamma^{(\gamma)})^{\!\top} & \Sigma_1
      \end{pmatrix}
   \Bigr),
\]
with cross--covariance $\Gamma^{(\gamma)}$ is (see \citet{janati2020entropic})
\begin{eqnarray*}
    \Gamma^{(\gamma)}
 &=&
    \frac{1}{2}\Sigma_0^{1/2}
      \Bigl(
        4\Sigma_0^{1/2}\Sigma_1\Sigma_0^{1/2}
        +\frac{\gamma^2}{4}\,I
      \Bigr)^{1/2}\Sigma_0^{-1/2}-\frac{\gamma}{4}\,I.
\end{eqnarray*}
As $\gamma\to0$ one recovers the deterministic (un-regularised) OT 
map with cross-covariance
$
\Gamma^{(0)}
  =\Sigma_0^{1/2}
     \bigl(\Sigma_0^{1/2}\Sigma_1\Sigma_0^{1/2}\bigr)^{1/2}
     \Sigma_0^{-1/2}.
$

\paragraph*{Penalized discrete matching}\label{app:pen:match}

Here, the classical entropic OT is
\[
\min_{P\in\mathcal{U}(\boldsymbol{1}_{n_0},\boldsymbol{1}_{n_1})}\sum_{i=1}^{n_1}\sum_{j=1}^{n_0}P_{ij}\,C_{ij}\;+\;
\gamma
\sum_{i,j} P_{ij}
\Bigl(
    \log\bigl(P_{ij}\bigr)-1
\Bigr)
\]
or equivalently,
\[
\min_{P\in\mathcal{U}(\boldsymbol{1}_{n_0},\boldsymbol{1}_{n_1})}\sum_{i=1}^{n_1}\sum_{j=1}^{n_0}P_{ij}\,C_{ij}\;+\;
\gamma
\sum_{i,j} P_{ij}
\Bigl(
    \log\bigl(n_0n_1P_{ij}\bigr)-1
\Bigr)
\]
and where $\mathcal{U}(\boldsymbol{1}_{n_1},\boldsymbol{1}_{n_0})$ is the polytope described previously. 
Define the Gibbs kernel
$
K=\exp(-C/\gamma)
$, element-wise (in the sense that the exponential function is applied to each entry). We seek scaling vectors
$\boldsymbol{u}\in\mathbb{R}^{n_{0}}_{+},\,\boldsymbol{v}\in\mathbb{R}^{n_{1}}_{+}$ satisfying
\[
\begin{cases}
\boldsymbol{u}\odot\bigl(K\boldsymbol{v}\bigr)=\displaystyle\frac{1}{n_{0}}\mathbf{1}_{n_{0}},\\
\boldsymbol{v}\odot\bigl(K^{\mathsf T}\boldsymbol{u}\bigr)=\displaystyle\frac{1}{n_{1}}\mathbf{1}_{n_{1}} ,
\end{cases}
\]
where $\odot$ is Hadamard product (element pointwise), and then the {entropic OT plan} has the popular Sinkhorn form
\[
P^{(\gamma)}= \operatorname{diag}(\boldsymbol{u})\,K\,\operatorname{diag}(\boldsymbol{v}).
\]
Note that one could also consider a proxy map, given by
\[
T^{(\gamma)}(x_i)=n_{0}\,\sum_{j=1}^{n_{1}} P_{ij}^{(\gamma)}\,y_j,
\quad
i=1,\dots,n_{0}.
\]
As $\gamma\!\to\!0$, $P^{(\gamma)}\!\to\!P^{\star}$ and 
$T^{(\gamma)}$ converges to the deterministic Monge map.

\subsection{Penalized matching with weights}\label{app:pen:match:weight}

In the generalized version of this problem, over the set of matrices, $\mathcal{U}(\boldsymbol{w}_0,\boldsymbol{w}_1)$, defined as
\[
\bigl\{\,P\in\mathbb{R}_+^{n_0\times n_1}:
P\,\mathbf{1}_{n_1}=\boldsymbol{w}_0,\ 
P^\top\mathbf{1}_{n_0}=\boldsymbol{w}_1
\bigr\},
\]
and consider
\[
\min_{P\in\mathcal{U}(\boldsymbol{w}_0,\boldsymbol{w}_1)}\sum_{i=1}^{n_1}\sum_{j=1}^{n_0}P_{ij}\,C_{ij}\;+\;
\gamma
\sum_{i,j} P_{ij}\log\!\Bigl(\frac{P_{ij}}{w_{1,i}\,w_{0,j}}\Bigr)
\]
We seek here scaling vectors $\boldsymbol{u}\in\mathbb{R}^{n_{0}}_{+}$, $\boldsymbol{v}\in\mathbb{R}^{n_{1}}_{+}$ such that
\[
\begin{cases}
\boldsymbol{u}\odot (K\boldsymbol{v})=\boldsymbol{w}_0,\\
\boldsymbol{v}\odot (K^{\top}\boldsymbol{u})=\boldsymbol{w}_1,
\end{cases}
\]
solved via the iterative scheme
\[
\boldsymbol{u}^{(t+1)}=\frac{\boldsymbol{w}_0}{K\boldsymbol{v}^{(t)}},\text{ and }
\boldsymbol{v}^{(t+1)}=\frac{\boldsymbol{w}_1}{K^{\top}\boldsymbol{u}^{(t+1)}}.
\]
Then, we have again
\[
P^{(\gamma)}=\operatorname{diag}(\boldsymbol{u})\,K\,\operatorname{diag}(\boldsymbol{v})
\]
is the {entropic OT plan}.  
As $\gamma\!\to\!0$, $P^{(\gamma)}\!\to\!P^{\star}$ solving the unpenalized optimization problem, as discussed in \cite{peyre2019computational}.
Here again, a deterministic proxy mapping for $P^{(\gamma)}$ is obtained by
\[
T^{(\gamma)}(x_i)
=\frac{1}{w_{0,i}}\sum_{j=1}^{n_{1}} P^{(\gamma)}_{ij}\,y_j,
\qquad i=1,\dots,n_{0}.
\]
It converges to the true Monge map (when it exists) in the limit $\gamma\!\to\!0$.
\subsection{Univariate conditional transport (and the ST two-parent-value notation)}
\label{app:conditional}

This appendix recalls the univariate monotone OT map and explains how it extends to conditional transport,
using the same notation as in the main text.

\paragraph*{Unconditional (one-dimensional) transport.}
Let $X\mid A=a \sim \nu_a$ be a univariate random variable with CDF $F_a$ and quantile function $Q_a$.
The unique monotone OT map from $\nu_0$ to $\nu_1$ is
\[
T(x)=Q_1(F_0(x)).
\]
In the Gaussian case, if $X\mid A=a\sim \mathcal N(\mu_{a},\sigma_{a}^2)$, then
\[
T(x)=\mu_{1}+\frac{\sigma_{1}}{\sigma_{0}}(x-\mu_{0}).
\]

\paragraph*{Conditional transport and the $(\boldsymbol z_0,\boldsymbol z_1)$ notation.}
Let $X_j$ be a mediator with parent set $\boldsymbol Z_j=\mathrm{parents}_{\mathcal G}(X_j)$.
Write $F_{a,j}(\cdot\mid \boldsymbol z)$ and $Q_{a,j}(\cdot\mid \boldsymbol z)$ for the conditional CDF and quantile
of $X_j$ given $\boldsymbol Z_j=\boldsymbol z$ in group $A=a$.
The conditional monotone OT map from the law of $X_j\mid (A=0,\boldsymbol Z_j=\boldsymbol z_0)$
to the law of $X_j\mid (A=1,\boldsymbol Z_j=\boldsymbol z_1)$ is
\[
T_j(x\mid \boldsymbol z_0,\boldsymbol z_1)
\;=\;
Q_{1,j}\!\left(F_{0,j}(x\mid \boldsymbol z_0)\mid \boldsymbol z_1\right).
\]
In Section~5, we distinguish $\boldsymbol z_0$ (the parent values used on the {\em source} side for $F_{0,j}$)
from $\boldsymbol z_1$ (the parent values used on the {\em target} side for $Q_{1,j}$) to make explicit the sequential
nature of ST. The commonly used shorthand
\[
T_{j\mid \boldsymbol Z_j}(x\mid \boldsymbol z)
\;:=\;
Q_{1,j}\!\left(F_{0,j}(x\mid \boldsymbol z)\mid \boldsymbol z\right)
\]
corresponds to the special case $\boldsymbol z_0=\boldsymbol z_1=\boldsymbol z$.

\paragraph*{Bivariate illustration (a chain $X_1\to X_2$).}
Consider $(X_1,X_2)$ with $X_1$ a parent of $X_2$.
For a unit in group $A=0$ with observed values $(x_1,x_2)$, ST first transports $x_1$ via
\[
x_1^\dagger \;=\; T_1(x_1)=Q_{1,1}(F_{0,1}(x_1)),
\]
and then transports $x_2$ {\em conditionally}, using the observed parent value $x_1$ on the source side and the transported
parent value $x_1^\dagger$ on the target side:
\[
T_2(x_2\mid x_1,x_1^\dagger)
\;=\;
Q_{1,2}\!\left(F_{0,2}(x_2\mid x_1)\mid x_1^\dagger\right).
\]
This is precisely the ST logic: conditioning for the target quantile is performed at the already transported parent value.

\paragraph*{Gaussian closed form.}
Assume $\boldsymbol X=(X_1,X_2)^\top$ satisfies $\boldsymbol X\mid A=a\sim\mathcal N(\boldsymbol\mu_a,\boldsymbol\Sigma_a)$ with
\[
\boldsymbol\mu_a=
\begin{pmatrix}\mu_{1,a}\\ \mu_{2,a}\end{pmatrix},
\qquad
\boldsymbol\Sigma_a=
\begin{pmatrix}
\sigma_{1,a}^2 & \rho_a\sigma_{1,a}\sigma_{2,a}\\
\rho_a\sigma_{1,a}\sigma_{2,a} & \sigma_{2,a}^2
\end{pmatrix}.
\]
Then $X_2\mid (A=a,X_1=x)$ is Gaussian with
\[
\mu_{2\mid a,x}=\mu_{2,a}+\rho_a\frac{\sigma_{2,a}}{\sigma_{1,a}}(x-\mu_{1,a}),
\qquad
\sigma^2_{2\mid a,x}=\sigma_{2,a}^2(1-\rho_a^2),
\]
so the conditional monotone OT map from $X_2\mid (A=0,X_1=x_1)$ to $X_2\mid (A=1,X_1=x_1^\dagger)$ is
\[
T_2(x_2\mid x_1,x_1^\dagger)
=
\mu_{2\mid 1,x_1^\dagger}
+
\frac{\sigma_{2\mid 1,x_1^\dagger}}{\sigma_{2\mid 0,x_1}}
\bigl(x_2-\mu_{2\mid 0,x_1}\bigr),
\]
with $x_1^\dagger=T_1(x_1)=\mu_{1,1}+\frac{\sigma_{1,1}}{\sigma_{1,0}}(x_1-\mu_{1,0})$.

\subsection{Categorical transport via simplex OT}\label{app:cat:ot}
Let a categorical mediator $X\in\{1,\dots,K\}$ have conditional probability vector
$p_a(z)=(p_{a,1}(z),\dots,p_{a,K}(z))\in\mathcal S_K$, where $p_{a,k}(z)=\mathbb P(X=k\mid Z=z,A=a)$.
We estimate $p_a(z)$ using a multiclass classifier, yielding $\widehat p_a(z)\in\mathcal S_K$.

We then define a deterministic transport map $T_{\mathcal S}:\mathcal S_K\to\mathcal S_K$
as the optimal simplex transport under the squared Euclidean cost on $\mathcal S_K$
(as in \cite{machado2025categorical}). In practice, $T_{\mathcal S}$ may be computed
via an entropically regularized OT solver (Sinkhorn) for numerical stability; the entropic
term is a regularizer and the underlying cost remains quadratic.

Finally, we obtain a categorical counterfactual assignment by
\[
X^\star(z) = \arg\max_{k\in[K]} \bigl[T_{\mathcal S}(p_0(z))\bigr]_k,
\]
and we apply the same construction conditionally on transported parent values in the sequential
procedure.

\paragraph*{Semi-discrete allocation to enforce category marginals.}
The argmax rule above yields a deterministic label, but it does not guarantee that the
resulting empirical category frequencies match the target marginals under $A=1$.
When we wish to enforce these marginals (conditionally on $Z=z$), we replace the argmax
rounding by a semi-discrete OT {\em allocation} from transported simplex vectors to the
simplex vertices.
Concretely, let $\tilde p_i = T_{\mathcal S}(\hat p_0(z_i)) \in \mathcal S_K$ for the units
in the source group, and let $\pi_k(z) = \mathbb P(X=k \mid Z=z, A=1)$ denote the target
(category) proportions. Consider the semi-discrete OT problem between the empirical measure
$\widehat{\nu} = \frac{1}{n}\sum_{i=1}^n \delta_{\tilde p_i}$ on $\mathcal S_K$ and the discrete
measure $\nu^\pi = \sum_{k=1}^K \pi_k(z)\,\delta_{e_k}$ supported on the vertices
$\{e_1,\dots,e_K\}$, with cost $c(p,e_k)=\|p-e_k\|^2$.
The optimal transport induces a partition of $\mathcal S_K$ into (weighted) Voronoi/Laguerre
cells, yielding a deterministic map $C:\mathcal S_K\to\{1,\dots,K\}$ such that the pushforward
$C_\#\widehat{\nu}=\nu^\pi$.
We then set $X^\star(z_i)=C(\tilde p_i)$.
This is the allocation step summarized in Algorithm~\ref{alg:split} (Appendix~\ref{app:algorithms}).



\section{Proofs of theoretical results}
\label{app:proofs}

This appendix provides detailed proofs for the theoretical results in Section~\ref{sec:theory}. 
Throughout, $j$ is a mediator with parent set $\boldsymbol Z_j$, and we write
\[
\begin{cases}
    T_j(x\mid \boldsymbol z_0,\boldsymbol z_1)=Q_{1,j}\!\left(F_{0,j}(x\mid \boldsymbol z_0)\mid \boldsymbol z_1\right),\\
\widehat T_j(x\mid \boldsymbol z_0,\boldsymbol z_1)=\widehat Q_{1,j}\!\left(\widehat F_{0,j}(x\mid \boldsymbol z_0)\mid \boldsymbol z_1\right).
\end{cases}
\]
When $\boldsymbol Z_j=\emptyset$, we omit conditioning on $\boldsymbol z_0,\boldsymbol z_1$.
We use the shorthand $o_p(1)$ for a sequence converging to $0$ in probability.

\subsection{Proof of Lemma~\ref{lem:base-root}}

\begin{proof}Let $\mathcal X_j$ be a compact subset of the interior of $\mathrm{supp}(X_j\mid A=0)$.
By continuity and strict monotonicity of $F_{0,j}$ on the support (Assumptions~\ref{A1}--\ref{A2}),
the image $F_{0,j}(\mathcal X_j)$ is a compact subset of $(0,1)$.
Hence there exist $0<\underline u<\overline u<1$ such that
$F_{0,j}(x)\in[\underline u,\overline u]$ for all $x\in\mathcal X_j$.

For any $x\in\mathcal X_j$, decompose
\[
\bigl|\widehat T_j(x)-T_j(x)\bigr|
=
\bigl|\widehat Q_{1,j}(\widehat F_{0,j}(x)) - Q_{1,j}(F_{0,j}(x))\bigr|
\le (I)_x + (II)_x,
\]
where
\[
(I)_x := \bigl|\widehat Q_{1,j}(\widehat F_{0,j}(x)) - Q_{1,j}(\widehat F_{0,j}(x))\bigr|,
\qquad
(II)_x := \bigl|Q_{1,j}(\widehat F_{0,j}(x)) - Q_{1,j}(F_{0,j}(x))\bigr|.
\]
Taking suprema over $x\in\mathcal X_j$ yields
\[
\sup_{x\in\mathcal X_j}(I)_x \le \sup_{u\in(0,1)}|\widehat Q_{1,j}(u)-Q_{1,j}(u)| = o_p(1)
\]
by Assumption~\ref{A3}. Also, Assumption~\ref{A3} gives
$\Delta_n := \sup_{x\in\mathcal X_j}|\widehat F_{0,j}(x)-F_{0,j}(x)|=o_p(1)$.
Since $Q_{1,j}$ is monotone and continuous on $(0,1)$ (Assumptions~\ref{A1}--\ref{A2}),
it is uniformly continuous on $[\underline u/2,1-\underline u/2]$; thus
\[
\sup_{x\in\mathcal X_j}(II)_x
\le \omega_{Q_{1,j}}(\Delta_n) = o_p(1),
\]
where $\omega_{Q_{1,j}}(\cdot)$ is a modulus of continuity of $Q_{1,j}$ on that compact interval.
Combining the bounds proves $\displaystyle\sup_{x\in\mathcal X_j}|\widehat T_j(x)-T_j(x)|\xrightarrow{p}0$.
\end{proof}

\subsection{Proof of Proposition~\ref{prop:cv:Tx}}

\begin{proof}
Fix $(x,\boldsymbol z_0,\boldsymbol z_1)$ in the interior of the support.
By Assumption~\ref{A3}, $\widehat F_{0,j}(x\mid \boldsymbol z_0)\xrightarrow{p}F_{0,j}(x\mid \boldsymbol z_0)$
and $\displaystyle\sup_{u,\boldsymbol z}|\widehat Q_{1,j}(u\mid \boldsymbol z)-Q_{1,j}(u\mid \boldsymbol z)|\xrightarrow{p}0$.

Decompose
\begin{align*}
\left|\widehat T_j(x\mid \boldsymbol z_0,\boldsymbol z_1)-T_j(x\mid \boldsymbol z_0,\boldsymbol z_1)\right|
&=
\left|\widehat Q_{1,j}(\widehat F_{0,j}(x\mid \boldsymbol z_0)\mid \boldsymbol z_1)
-
Q_{1,j}(F_{0,j}(x\mid \boldsymbol z_0)\mid \boldsymbol z_1)\right|\\
&\le (I)+(II),
\end{align*}
with
\begin{align*}
(I)&:=\left|\widehat Q_{1,j}(\widehat F_{0,j}(x\mid \boldsymbol z_0)\mid \boldsymbol z_1)
-
Q_{1,j}(\widehat F_{0,j}(x\mid \boldsymbol z_0)\mid \boldsymbol z_1)\right|\\
&\le \sup_{u,\boldsymbol z}\left|\widehat Q_{1,j}(u\mid \boldsymbol z)-Q_{1,j}(u\mid \boldsymbol z)\right|=o_p(1),
\end{align*}
and
\[
(II):=\left|Q_{1,j}(\widehat F_{0,j}(x\mid \boldsymbol z_0)\mid \boldsymbol z_1)
-
Q_{1,j}(F_{0,j}(x\mid \boldsymbol z_0)\mid \boldsymbol z_1)\right|.
\]
By Assumptions~\ref{A1}--\ref{A2}, for fixed $\boldsymbol z_1$, the map $u\mapsto Q_{1,j}(u\mid \boldsymbol z_1)$
is continuous on $(0,1)$, hence $Q_{1,j}(\widehat F_{0,j}(x\mid \boldsymbol z_0)\mid \boldsymbol z_1)\xrightarrow{p}
Q_{1,j}(F_{0,j}(x\mid \boldsymbol z_0)\mid \boldsymbol z_1)$ by the continuous mapping theorem.
Thus $(II)=o_p(1)$, proving the claim.
\end{proof}

\subsection{Proof of Corollary~\ref{cor:cv:Tx-seq}}

\begin{proof}
This is an immediate consequence of Proposition~\ref{prop:cv:Tx} by evaluating it at the deterministic choice
$\boldsymbol z_1=T_{\boldsymbol Z_j}(\boldsymbol z_0)$.
\end{proof}

\subsection{Proof of Proposition~\ref{prop:5:2}}

\begin{proof}
Fix $\boldsymbol z_0$ and write $\boldsymbol z_1=T_{\boldsymbol Z_j}(\boldsymbol z_0)$ and
$\widehat{\boldsymbol z}_1=\widehat T_{\boldsymbol Z_j}(\boldsymbol z_0)$.
From the induction hypothesis \ref{A7}, $\widehat{\boldsymbol z}_1\xrightarrow{p}\boldsymbol z_1$, and let
$u_0:=F_{0,j}(x\mid \boldsymbol z_0)\in[\varepsilon,1-\varepsilon]$.

Decompose
\[
\bigl|\widehat T_j(x\mid \boldsymbol z_0,\widehat{\boldsymbol z}_1)-T_j(x\mid \boldsymbol z_0,\boldsymbol z_1)\bigr|
\le (I)+(II),
\]
where
\[
(I):=\bigl|\widehat T_j(x\mid \boldsymbol z_0,\widehat{\boldsymbol z}_1)-T_j(x\mid \boldsymbol z_0,\widehat{\boldsymbol z}_1)\bigr|,
\qquad
(II):=\bigl|T_j(x\mid \boldsymbol z_0,\widehat{\boldsymbol z}_1)-T_j(x\mid \boldsymbol z_0,\boldsymbol z_1)\bigr|.
\]

For $(I)$, write, for $\widehat u_0:=\widehat F_{0,j}(x\mid \boldsymbol z_0)$,
\[
(I)\le 
\underbrace{\bigl|\widehat Q_{1,j}(\widehat u_0\mid \widehat{\boldsymbol z}_1)-Q_{1,j}(\widehat u_0\mid \widehat{\boldsymbol z}_1)\bigr|}_{(I.a)}
+
\underbrace{\bigl|Q_{1,j}(\widehat u_0\mid \widehat{\boldsymbol z}_1)-Q_{1,j}(u_0\mid \widehat{\boldsymbol z}_1)\bigr|}_{(I.b)},
\]
By Assumption~\ref{A3}, $(I.a)\le \displaystyle\sup_{u,\boldsymbol z}|\widehat Q_{1,j}(u\mid \boldsymbol z)-Q_{1,j}(u\mid \boldsymbol z)|=o_p(1)$.
Also $\widehat u_0\xrightarrow{p}u_0$ by Assumption~\ref{A3}; moreover, since $u_0\in[\varepsilon,1-\varepsilon]$,
we have $\widehat u_0\in[\varepsilon/2,1-\varepsilon/2]$ with probability tending to $1$.
By continuity of $u\mapsto Q_{1,j}(u\mid \widehat{\boldsymbol z}_1)$ on $(0,1)$ (Assumptions~\ref{A1}--\ref{A2}),
this implies $(I.b)=o_p(1)$. Hence $(I)=o_p(1)$.

For $(II)$, note that
\[
T_j(x\mid \boldsymbol z_0,\boldsymbol z)=Q_{1,j}(u_0\mid \boldsymbol z)
\]
with $u_0\in[\varepsilon,1-\varepsilon]$. Therefore, by Assumption~\ref{A6} (uniform continuity in $\boldsymbol z$ uniformly over
$u\in[\varepsilon,1-\varepsilon]$),
\[
(II)=\bigl|Q_{1,j}(u_0\mid \widehat{\boldsymbol z}_1)-Q_{1,j}(u_0\mid \boldsymbol z_1)\bigr|\xrightarrow{p}0
\quad\text{since }\widehat{\boldsymbol z}_1\xrightarrow{p}\boldsymbol z_1.
\]
Combining $(I)=o_p(1)$ and $(II)=o_p(1)$ proves the claim.
\end{proof}

\subsection{Proof of Proposition~\ref{prop:cat:consistency:pointwise}}

\begin{proof}
Fix $(\boldsymbol z_0,\boldsymbol z_1)$ in the interior of $\mathrm{supp}(\boldsymbol Z_j)$.
In the categorical case, the conditional laws are probability vectors
$p_{a,j}(\cdot\mid \boldsymbol z)$ on a finite support.
Let $\widehat p_{a,j}(\cdot\mid \boldsymbol z)$ be their estimators.

By Assumption~\ref{A1} (finite support) together with the assumed consistency of the conditional probability estimators
(e.g.,~\ref{C2}), we have
$\widehat p_{a,j}(\cdot\mid \boldsymbol z)\xrightarrow{p} p_{a,j}(\cdot\mid \boldsymbol z)$ (componentwise) at the fixed covariate values
$\boldsymbol z_0,\boldsymbol z_1$.
By stability/continuity of the simplex transport operator at interior points (e.g.,~\ref{C1} or the corresponding result in
\citet{machado2025categorical}), the plug-in simplex transport estimator satisfies
\[
\widehat T_j\bigl(p_{0,j}(\cdot\mid \boldsymbol z_0)\mid \boldsymbol z_0,\boldsymbol z_1\bigr)
\xrightarrow{p}
T_j\bigl(p_{0,j}(\cdot\mid \boldsymbol z_0)\mid \boldsymbol z_0,\boldsymbol z_1\bigr),
\]
which is the desired pointwise convergence.
\end{proof}

\subsection{Proof of Corollary~\ref{cor:cat:consistency:path}}

\begin{proof}
Immediate from Proposition~\ref{prop:cat:consistency:pointwise} by evaluating it at the deterministic choice
$\boldsymbol z_1=T_{\boldsymbol Z_j}(\boldsymbol z_0)$.
\end{proof}

\subsection{Proof of Proposition~\ref{prop:cat:consistency:plugin}}

\begin{proof}
Fix $\boldsymbol z_0$ and write $\boldsymbol z_1=T_{\boldsymbol Z_j}(\boldsymbol z_0)$ and
$\widehat{\boldsymbol z}_1=\widehat T_{\boldsymbol Z_j}(\boldsymbol z_0)$.
Assume $\widehat{\boldsymbol z}_1\xrightarrow{p}\boldsymbol z_1$ (induction hypothesis).

Decompose
\[
\left\|
\widehat T_j\bigl(p_{0,j}(\cdot\mid \boldsymbol z_0)\mid \boldsymbol z_0,\widehat{\boldsymbol z}_1\bigr)
-
T_j\bigl(p_{0,j}(\cdot\mid \boldsymbol z_0)\mid \boldsymbol z_0,\boldsymbol z_1\bigr)
\right\|_\infty
\le (I)+(II),
\]
where
\[
(I):=
\left\|
\widehat T_j\bigl(p_{0,j}(\cdot\mid \boldsymbol z_0)\mid \boldsymbol z_0,\widehat{\boldsymbol z}_1\bigr)
-
T_j\bigl(p_{0,j}(\cdot\mid \boldsymbol z_0)\mid \boldsymbol z_0,\widehat{\boldsymbol z}_1\bigr)
\right\|_\infty,
\]
\[
(II):=
\left\|
T_j\bigl(p_{0,j}(\cdot\mid \boldsymbol z_0)\mid \boldsymbol z_0,\widehat{\boldsymbol z}_1\bigr)
-
T_j\bigl(p_{0,j}(\cdot\mid \boldsymbol z_0)\mid \boldsymbol z_0,\boldsymbol z_1\bigr)
\right\|_\infty.
\]

Term $(I)$ vanishes by the pointwise consistency of the simplex transport estimator at fixed covariates
(Proposition~\ref{prop:cat:consistency:pointwise}) applied at $(\boldsymbol z_0,\widehat{\boldsymbol z}_1)$
together with a standard “random index” argument; a sufficient condition is local uniformity/stochastic equicontinuity
of $\widehat T_j(\cdot\mid \boldsymbol z_0,\boldsymbol z_1)$ in $\boldsymbol z_1$ on compacts (which holds under the stability results for
discrete OT on the simplex; see \citet{machado2025categorical}).

For $(II)$, by continuity of the population simplex transport operator in the covariate argument
(e.g., induced by continuity of $\boldsymbol z\mapsto p_{1,j}(\cdot\mid \boldsymbol z)$ and stability of the OT solution on the simplex,
cf.~\ref{C1}), and since $\widehat{\boldsymbol z}_1\xrightarrow{p}\boldsymbol z_1$, we obtain $(II)=o_p(1)$.

Hence $(I)+(II)=o_p(1)$, proving the claim.
\end{proof}

\subsection{Proof of Theorem~\ref{thm:seq:consistency}}

\begin{proof}
We argue by induction along the chosen topological order $\pi$. 
Recall \ref{A7} from Section~\ref{sec:5:2} (applied here with $j=\pi(m-1)$), i.e.,
$x^\star_{1,i,\pi(\ell)}\to x^\dagger_{1,i,\pi(\ell)}$ for all $\ell\le m-1$.

{\em Base step.}
For the first node $j=\pi(1)$, which is necessarily a root of the DAG, the transported value is
$x^\star_{1,i,j}=\widehat T_j(x_{0,i,j})$.
Consistency follows from Lemma~\ref{lem:base-root} when $X_j$ is numerical, and from the corresponding
root-level plug-in consistency result for the simplex transport when $X_j$ is categorical.

{\em Induction step.}
Fix $m\in\{2,\ldots,d\}$ and set $j=\pi(m)$. Let
\[
\boldsymbol z_{0,i,j}:=\boldsymbol x_{0,i,\mathrm{parents}_{\mathcal{G}}(X_j)},
\qquad
\boldsymbol z_{1,i,j}:=\boldsymbol x^\dagger_{1,i,\mathrm{parents}_{\mathcal{G}}(X_j)},
\qquad
\widehat{\boldsymbol z}_{1,i,j}:=\boldsymbol x^\star_{1,i,\mathrm{parents}_{\mathcal{G}}(X_j)} .
\]
By the induction hypothesis, for every parent $\ell\in \mathrm{parents}_{\mathcal{G}}(X_j)$ we have
$x^\star_{1,i,\ell}\xrightarrow{p}x^\dagger_{1,i,\ell}$, and therefore
\[
\widehat{\boldsymbol z}_{1,i,j}=\boldsymbol x^\star_{1,i,\mathrm{parents}_{\mathcal{G}}(X_j)}
\xrightarrow{p}
\boldsymbol x^\dagger_{1,i,\mathrm{parents}_{\mathcal{G}}(X_j)}=\boldsymbol z_{1,i,j}.
\]
We now apply the appropriate plug-in consistency result at node $j$.
If $X_j$ is numerical, Proposition~\ref{prop:5:2} yields
\[
x^\star_{1,i,j}
=
\widehat T_j\!\left(x_{0,i,j}\mid \boldsymbol z_{0,i,j},\widehat{\boldsymbol z}_{1,i,j}\right)
\xrightarrow{p}
T_j\!\left(x_{0,i,j}\mid \boldsymbol z_{0,i,j},\boldsymbol z_{1,i,j}\right)
=
x^\dagger_{1,i,j}.
\]
If $X_j$ is categorical, Proposition~\ref{prop:cat:consistency:plugin} yields the same conclusion.
Hence $x^\star_{1,i,\pi(m)}\xrightarrow{p}x^\dagger_{1,i,\pi(m)}$.

Since $m$ was arbitrary, the above holds for every coordinate along $\pi$, and thus
$\boldsymbol x^\star_{1,i}\xrightarrow{p}\boldsymbol x^\dagger_{1,i}$.
\end{proof}

\subsection{Proof of Lemma~\ref{cor:cat:assignment}}

\begin{proof}
Let $\widehat{\boldsymbol p}_n\in\mathcal S_K$ be the estimated transported probability vector and
$\boldsymbol p\in\mathcal S_K$ its population counterpart, with
$\|\widehat{\boldsymbol p}_n-\boldsymbol p\|_\infty\xrightarrow{p}0$
(as given by Proposition~\ref{prop:cat:consistency:plugin}).
Let $g:\mathcal S_K\to\{1,\dots,K\}$ be a deterministic label mapping.

Assume $g$ is locally constant at $\boldsymbol p$, i.e., there exists $\gamma>0$ such that
\[
\|\boldsymbol q-\boldsymbol p\|_\infty<\gamma \ \Longrightarrow\ g(\boldsymbol q)=g(\boldsymbol p).
\]
Then
\[
\mathbb P\bigl(g(\widehat{\boldsymbol p}_n)\neq g(\boldsymbol p)\bigr)
\le
\mathbb P\bigl(\|\widehat{\boldsymbol p}_n-\boldsymbol p\|_\infty\ge \gamma\bigr)\ \longrightarrow\ 0,
\]
since $\|\widehat{\boldsymbol p}_n-\boldsymbol p\|_\infty\xrightarrow{p}0$.
Hence $\mathbb P(g(\widehat{\boldsymbol p}_n)=g(\boldsymbol p))\to 1$.
\end{proof}

In the case of the argmax rule with a positive margin, let $k^\star=\arg\displaystyle\max_k p_k$ and assume the margin
\[
m:= p_{k^\star}-\max_{\ell\neq k^\star}p_\ell>0.
\]
On the event $\|\widehat{\boldsymbol p}_n-\boldsymbol p\|_\infty < m/2$, we have
$\widehat p_{n,k^\star} > \displaystyle\max_{\ell\neq k^\star}\widehat p_{n,\ell}$, hence
$\arg\displaystyle\max_k \widehat p_{n,k}=k^\star$ (and any fixed deterministic tie-break is irrelevant on this event).
Therefore
\[
\mathbb P\bigl(\arg\max_k \widehat p_{n,k}=k^\star\bigr)
\ge
\mathbb P\bigl(\|\widehat{\boldsymbol p}_n-\boldsymbol p\|_\infty < m/2\bigr)\to 1.
\]

\subsection{Proof of Theorem \ref{cor:consistency:effects} }

\begin{proof}
Fix a unit $i$ with $A_i=0$. Let
\[
E_i := \{\boldsymbol x^\dagger_{1,i}\in \mathcal{X}_0\}.
\]
By Assumption~\ref{ass:overlap-outcome}, $\mathbb P(E_i^c\mid A_i=0)\le \eta$.

We prove that, for any $\epsilon>0$,
\[
\mathbb P\!\left(
\bigl\|(\widehat\delta_i,\widehat\zeta_i,\widehat\tau_i)-(\delta_i^\dagger,\zeta_i^\dagger,\tau_i^\dagger)\bigr\|>\epsilon
\ \middle|\ A_i=0\right)
\le \eta+o(1).
\]
Decompose
\begin{align*}
\mathbb P(\|\cdot\|>\epsilon\mid A_i=0)
&\le
\mathbb P(\|\cdot\|>\epsilon, E_i\mid A_i=0) + \mathbb P(E_i^c\mid A_i=0)\\
&\le
\mathbb P(\|\cdot\|>\epsilon, E_i\mid A_i=0) + \eta.
\end{align*}
It remains to show $\mathbb P(\|\cdot\|>\epsilon, E_i\mid A_i=0)=o(1)$.

Assume (as stated in the corollary) that $\boldsymbol x^\star_{1,i}\xrightarrow{p}\boldsymbol x^\dagger_{1,i}$
(Theorem~\ref{thm:seq:consistency}), and that the outcome regressions satisfy the following standard plug-in condition:
\begin{equation}
\label{eq:mu-unif}
\sup_{\boldsymbol x\in S_a}\bigl|\widehat\mu_a(\boldsymbol x)-\mu_a(\boldsymbol x)\bigr|\xrightarrow{p}0,
\qquad a\in\{0,1\},
\end{equation}
and $\mu_a$ is continuous on $S_a$ (or at least continuous at the relevant evaluation points).
(Condition~\eqref{eq:mu-unif} can be replaced by stochastic equicontinuity plus pointwise consistency.)

On the event $E_i$, both $\boldsymbol x^\dagger_{1,i}$ and, with probability tending to $1$, $\boldsymbol x^\star_{1,i}$
belong to $\mathcal{X}_0$ (since $\boldsymbol x^\star_{1,i}\to \boldsymbol x^\dagger_{1,i}$ and $\mathcal{X}_0$ is a support set).
Then
\[
\widehat\mu_0(\boldsymbol x^\star_{1,i})-\mu_0(\boldsymbol x^\dagger_{1,i})
=
\underbrace{\bigl[\widehat\mu_0(\boldsymbol x^\star_{1,i})-\mu_0(\boldsymbol x^\star_{1,i})\bigr]}_{(A)}
+
\underbrace{\bigl[\mu_0(\boldsymbol x^\star_{1,i})-\mu_0(\boldsymbol x^\dagger_{1,i})\bigr]}_{(B)}.
\]
By \eqref{eq:mu-unif}, $(A)\le \sup_{\boldsymbol x\in \mathcal{X}_0}|\widehat\mu_0(\boldsymbol x)-\mu_0(\boldsymbol x)|=o_p(1)$.
By continuity of $\mu_0$ and $\boldsymbol x^\star_{1,i}\xrightarrow{p}\boldsymbol x^\dagger_{1,i}$, $(B)=o_p(1)$.
Hence
\[
\widehat\mu_0(\boldsymbol x^\star_{1,i})\xrightarrow{p}\mu_0(\boldsymbol x^\dagger_{1,i})
\quad\text{on }E_i.
\]
Similarly, $\widehat\mu_0(\boldsymbol x_{0,i})\xrightarrow{p}\mu_0(\boldsymbol x_{0,i})$ by consistency of $\widehat\mu_0$
(at a fixed point, this follows from \eqref{eq:mu-unif} or from pointwise consistency).
Therefore,
\[
\widehat\delta_i
=
\widehat\mu_0(\boldsymbol x^\star_{1,i})-\widehat\mu_0(\boldsymbol x_{0,i})
\xrightarrow{p}
\mu_0(\boldsymbol x^\dagger_{1,i})-\mu_0(\boldsymbol x_{0,i})
=\delta_i^\dagger
\quad\text{on }E_i.
\]
The same argument gives $\widehat\mu_1(\boldsymbol x^\star_{1,i})\xrightarrow{p}\mu_1(\boldsymbol x^\dagger_{1,i})$
(on the appropriate support $\mathcal{X}_1$; for ST transports, $\boldsymbol x^\dagger_{1,i}$ lies in the target support by construction),
and thus
\[
\widehat\zeta_i
=
\widehat\mu_1(\boldsymbol x^\star_{1,i})-\widehat\mu_0(\boldsymbol x^\star_{1,i})
\xrightarrow{p}
\mu_1(\boldsymbol x^\dagger_{1,i})-\mu_0(\boldsymbol x^\dagger_{1,i})
=\zeta_i^\dagger
\quad\text{on }E_i.
\]
Finally, $\widehat\tau_i=\widehat\delta_i+\widehat\zeta_i\xrightarrow{p}\delta_i^\dagger+\zeta_i^\dagger=\tau_i^\dagger$
on $E_i$ by Slutsky's theorem, hence
\[
\mathbb P(\|(\widehat\delta_i,\widehat\zeta_i,\widehat\tau_i)-(\delta_i^\dagger,\zeta_i^\dagger,\tau_i^\dagger)\|>\epsilon,\ E_i\mid A_i=0)=o(1).
\]
Plugging into the initial decomposition yields the claimed bound $\le \eta+o(1)$.
When $\eta=0$, the bound implies the stated convergence in probability.
\end{proof}

\section{Algorithms}
\label{app:algorithms}

We summarize the computational components of the proposed framework through three algorithms. Algorithm~\ref{alg:1} describes the {\em (conditional) sequential transport} (ST) procedure, which constructs mediator counterfactuals by transporting each mediator along a topological ordering of the mediator DAG. At each step, the algorithm applies either an unconditional or a conditional transport map, depending on the presence of parent mediators, and produces a deterministic counterfactual mediator vector that respects both the observed distributions and the causal structure.
Algorithm~\ref{alg:2} details the estimation of the transport maps used within ST. It covers both numerical mediators, for which transport is implemented via conditional quantile mapping, and categorical mediators, for which transport is performed on the probability simplex. These maps constitute the core building blocks of the sequential procedure and are estimated once from the observed data under each treatment level.
Finally, Algorithm~\ref{alg:split} combines sequential transport with outcome regression to compute individual-level and aggregate causal decompositions. Given fitted outcome models, it evaluates direct and indirect effects by contrasting factual and transported mediator profiles. Together, the three algorithms form a modular pipeline: Algorithm~\ref{alg:2} estimates the transport primitives, Algorithm~\ref{alg:1} constructs counterfactual mediators, and Algorithm~\ref{alg:split} translates these counterfactuals into effect decompositions.

\begin{algorithm}[t!]
\caption{Conditional transport for a numeric feature for $x_{0,j}$, from \cite{machado2025sequential}.}\label{alg:1}
\begin{algorithmic}
\Require observation $x_{0,j}\in\mathbb{R}$ (numeric)
\Require sample $\chi_{0,j}=\{x_{0,j,1},\cdots,x_{0,j,n_0}\}\in\mathbb{R}^{n_0}$ 
\Require sample $\chi_{1,j}=\{x_{1,j,1},\cdots,x_{1,j,n_1}\}\in\mathbb{R}^{n_1}$ 
\Require weights $\boldsymbol{w}_0\in\mathbb{R}_+^{n_0}$ and  $\boldsymbol{w}_1\in\mathbb{R}_+^{n_1}$
\State $F(\cdot)\gets$ smooth cdf from $(\chi_{0,j},\boldsymbol{w}_0)$
\State $Q(\cdot)\gets$ smooth quantile function from $(\chi_{1,j},\boldsymbol{w}_1)$\\
\Return $x^\star_{1,j}\gets Q\circ F(x_{0,j})$
\end{algorithmic}
\end{algorithm}

\begin{algorithm}[t!]
\caption{Conditional transport on causal graph for a categorical feature, from \cite{machado2025categorical}.}\label{alg:2}
\begin{algorithmic}
\Require observation $x_{0,j}\in\{A_1,\cdots,A_{d_j}\}$ (categorical)
\Require  $\chi_{0,j}=\{\widehat{p}_j(x_{0,j,1}),\cdots,\widehat{p}_j(x_{0,j,n_0)}\}\in\mathcal{S}_{d_j-1}^{n_0}$ 
\Require  $\chi_{1,j}=\{\widehat{p}_j(x_{1,j,1}),\cdots,\widehat{p}_j(x_{1,j,n_1)}\}\in\mathcal{S}_{d_j-1}^{n_1}$ 
\Require weights $\boldsymbol{w}_0\in\mathbb{R}_+^{n_0}$ and  $\boldsymbol{w}_1\in\mathbb{R}_+^{n_1}$
\Require categorical allocation $C_j$ (Algorithm \ref{alg:split})
\State $T^\star:\mathcal{S}_{d_j-1}\to\mathcal{S}_{d_j-1}$ optimal weighted transport 
\State $p^\star_{1,j}\gets T^\star(\widehat{p}_j(x_{0,j}))$\\
\Return $x^\star_{1,j}\gets C_j(p^\star_{1,j})\in\{1,\cdots,d_j\}$
\end{algorithmic}
\end{algorithm}

\begin{algorithm}[t]
\caption{From compositional to categorical feature, with marginal constraints.}\label{alg:split}
\begin{algorithmic}
\Require $n$ observations on $\mathcal{S}_{d-1}$, $\widehat{\boldsymbol{p}}_1,\cdots,\widehat{\boldsymbol{p}}_n$
\Require target proportions $\pi_1,\cdots,\pi_d$
\State Optimal mapping  $T^\star:\mathcal{S}_{d-1}\to\{1,\cdots,d\}$\\
(from $\widehat{\boldsymbol{p}}\in\mathcal{S}_{d-1}$ to vertices of $\mathcal{S}_{d-1}$, with masses $\pi$)\\
\Return $T^\star:\mathcal{S}_{d-1}\to\{1,\cdots,d\}$ (Voronoi tesselation)
\end{algorithmic}
\end{algorithm}

\section{Toy examples with Gaussian data}

\subsection{Toy example with 2 mediators}\label{app:toy}

Consider the following simple example: treatment $A$ is a binary variable, taking values in $\{0,1\}$ (untreated, treated) with probabilities $p_0$ and $p_1$. Then, conditional on $A$, $\boldsymbol{X}=(X_1,X_2)$ is a bivariate Gaussian vector, $\boldsymbol{X}\mid A=a\sim \mathcal{N}(\boldsymbol{\mu}_a,\boldsymbol{\Sigma}_a)$
\[
\boldsymbol{X}\mid A=a\sim \mathcal{N}\left(
\begin{pmatrix}
    \mu_a\\
    \mu_a
\end{pmatrix},
\begin{pmatrix}
    1 & r_a\\
    r_a & 1
\end{pmatrix}
\right)
\]

For numerical applications, we set $\mu_0=-1$ and $\mu_1=1$, and $r_a\in\{-0.5,+0.7\}$.
Then, $Y$ is generated based on a simple linear model, $$Y = 2 X_1 -1.5 X_2+3 A+\varepsilon$$ where $\varepsilon$ is $\mathcal{N}(0,1)$.
The two ``potential outcomes'' given by the SEM are:
$$
\begin{cases}
    Y(1) = 2 X_1 -1.5 X_2+3 +\varepsilon\\
    Y(0) = 2 X_1 -1.5 X_2+\varepsilon.
\end{cases}
$$
Under this setup, the expected values of the indirect effect $\bar{\delta}$, the direct effect $\bar{\zeta}$, and the total effect $\bar{\tau}$ are:
$$
\begin{cases}
    \bar{\delta} &= 0.5(\mu_1 - \mu_0) = 1\\
    \bar{\zeta} &= 3\\
    \bar{\tau} &= 0.5(\mu_1 - \mu_0) + 3=4.
\end{cases}
$$
We generate a sample of size $n = 500$ from this data-generating process (DGP) and construct counterfactuals for the mediators $\boldsymbol{X}$ using OT, entropy regularized OT, and ST. For the latter, we consider two factorizations: (i) transport $X_1$ first, then $X_1 \to X_2$ (denoted ST(1)); and (ii) transport $X_2$ first, then $X_2 \to X_1$ (denoted ST(2)).

Figure~\ref{fig:gaussian-transport-0to1} displays the original untreated units (green), treated units (yellow), and the transported units from $A = 0$ to $A = 1$ (triangles). The transport mappings for each method are indicated by line segments.

\begin{figure}[htb!]
    \centering
    \includegraphics[width=\linewidth]{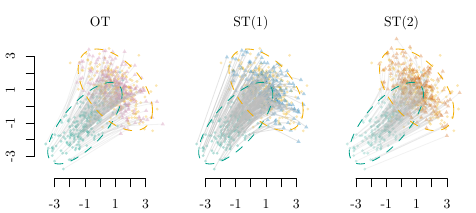}
    \caption{Untreated units (green dots), treated Units (yellow dots), and transported units (blue triangles) under optimal transport (OT), entropic regularization with Sinkhorn algorithm (SKH) and sequential transport: ST(1) (transporting $X_1$ first and then $X_1\to X_2$) and ST(2) (transporting $X_2$ first and then $X_2\to X_1$). Lines indicate the transport mappings.}
    \label{fig:gaussian-transport-0to1}
\end{figure}

\begin{figure}[t!]
    \centering
    \includegraphics[width=\linewidth]{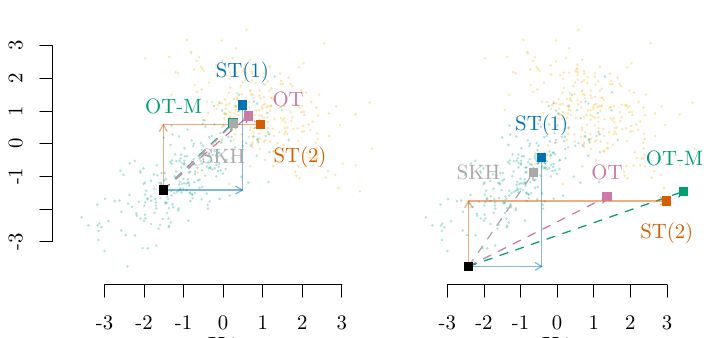}
    \caption{Estimated $\tau_i$ for the two untreated units (black squares), $\boldsymbol{x}_{i_1} = (-1.5, -1.4)$ (left) and $\boldsymbol{x}_{i_2} = (-2.4, -3.8)$ (right), using different transport methods. Green dots: untreated; yellow dots: treated.}
    \label{fig:gaussian-tau-two-indiv}
\end{figure}

For individual effects estimations, we focus on two individuals, $\boldsymbol{x}_{i_1} = (-1.5, -1.4)$ and $\boldsymbol{x}_{i_2} = (-2.4, -3.8)$ (Figure~\ref{fig:gaussian-tau-two-indiv}). For each of them, we build counterfactuals using transport-based methods. Using these counterfactuals, we compute the individual direct, indirect and total causal effects. The resulting values are reported in Table~\ref{tab:individual-causal-effects-toy}.

\begin{table}[tb!]
    \centering
     \caption{Estimated individual causal effects for two units, $\boldsymbol{x}_{i_1} = (-1.5, -1.4)$ and $\boldsymbol{x}_{i_2} = (-2.4, -3.8)$, based on counterfactuals build using optimal transport (OT), OT-Based matching (OT-M), penalized OT (SKH), and two sequential appoaches $X_1 \rightarrow X_2$ (ST(1)) and $X_2 \rightarrow X_1$ (ST(2)).}
    \label{tab:individual-causal-effects-toy}
    \begin{tabular}[t]{rrrrrrr}
\toprule
Metric & i & OT & OT-M & SKH & ST(1) & ST(2)\\
\midrule
\multirow{2}*{$\delta_i$} & 1 & 1.3 & 0.4 & 0.5 & -0.4 & 1.4\\
& 2 & 4.2 & 4.1 & 1.0 & 1.1 & 3.9\\
\cmidrule(lr){1-7}
\multirow{2}*{$\zeta_i$} & 1 & 2.7 & 3.2 & 3.2 & 3.7 & 3.5\\
 & 2 & 4.5 & 7.7 & 3.4 & 3.3 & 7.6\\
 \cmidrule(lr){1-7}
\multirow{2}*{$\tau_i$} & 1 & 4.0 & 3.5 & 3.7 & 3.3 & 4.9\\
 & 2 & 8.8 & 11.8 & 4.4 & 4.4 & 11.9\\
\bottomrule
\end{tabular}
\end{table}


We repeat the procedure over 200 Monte Carlo replications to evaluate performance. In each replication, we generate a new sample of size $n = 500$, estimate counterfactuals using OT and ST, and compute both individual and average causal effects. The results of the estimations are presented in Section~\ref{sec:6:gaussian-sims}.

In addition to the example with the two untreated units showcased in the main analysis, 
we can complement the results by having a look at a generalization for the entire set of untreated units. To that end, the distribution of individual causal effects obtained using each of the counterfactual-based methods is plotted in Figure~\ref{fig:gaussian-indiv-effects}. As shown, the distributions of estimated individual effects are centered around their corresponding theoretical expectations.

\begin{figure}[!tb]
    \centering
    \includegraphics[width=\linewidth]{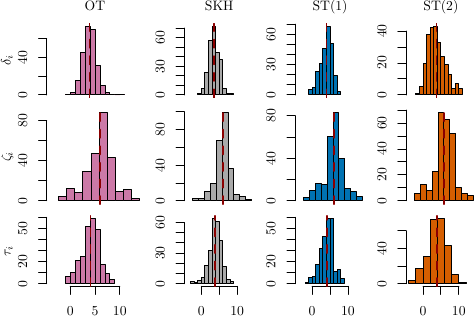}
    \caption{Individual causal effects for the Gaussian toy example estimated with counterfactual-based methods using optimal transport (OT), penalized OT (SKH), and two sequential approaches $X_1 \rightarrow X_2$ (ST(1)) and $X_2 \rightarrow X_1$ (ST(2)). The red dashed line depicts theoretical expectations.}
    \label{fig:gaussian-indiv-effects}
\end{figure}

\subsubsection*{Varying the means}\label{app:toy-alpha}

We change the distance between the means of the Gaussian distributions in groups~0 and~1. As this distance increases, the overlap between the two groups decreases. In the initial setup, with $\boldsymbol{\mu}_0 = -1$ and $\boldsymbol{\mu}_1 = +1$, the Euclidean distance is equal to $2 \sqrt{2}$. We apply a scalar coefficient $\alpha\geq 0$ to the means to increase that distance: $\alpha\boldsymbol{\mu}_0$ and $\alpha\boldsymbol{\mu}_1$. We make $\alpha$ vary from 0 to 2 by steps of 0.2. When $\alpha=0$, the distance is equal to $0$, when $\alpha=2$, the distance is equal to $4\sqrt{2}$. A visual representation of the samples that can be drawn from the DGP with $\alpha \in\{0,1,2\}$ is shown in Figure~\ref{fig:gaussian-example-alpha}.


\begin{figure}[htb!]
    \centering
    \includegraphics[width=\linewidth]{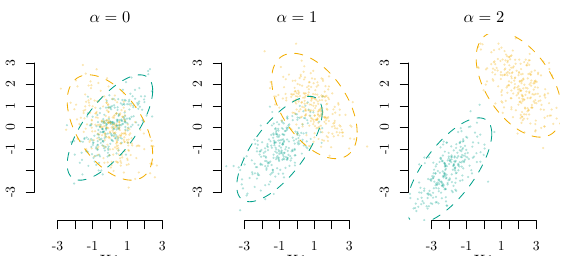}
    \caption{Example of simulated data for the toy example of section~\ref{app:toy} for various distances between the means of the Gaussian distributions in groups 0 and 1.}
    \label{fig:gaussian-example-alpha}
\end{figure}

We perform 200 Monte Carlo replications for each value of $\alpha$. For each replication, we compute the estimated average causal effects ($\hat{\bar\delta}$, $\hat{\bar\zeta}$, and $\hat{\bar\tau}$) and evaluate their estimation errors with respect to their theoretical values:
$$
\begin{aligned}
    \hat{\bar\delta} - \bar{\delta} &= \hat{\bar\delta} - \alpha, \\
    \hat{\bar\zeta} - \bar{\zeta} &= \hat{\bar\zeta} - 3, \\
    \hat{\bar\tau} - \bar\tau &= \hat{\bar\tau} - (\alpha + 3).
\end{aligned}
$$

The distributions of these estimation errors, obtained across the replications, are visualized as violin plots in Figure~\ref{fig:gaussian-mc-alpha} for the indirect effect (top), the direct effect (middle), and the total effect (bottom). The methods to compute these quantities are shown in columns. For small values of $\alpha$, corresponding to minimal separation between the two treatment groups, the indirect effect $\hat{\bar\delta}$ is slightly overestimated by transport-based methods. As $\alpha$ increases, the variance of the estimation error grows for the causal mediation (CM) method. Methods based on transport-based counterfactuals remain approximately unbiased for all values of $\alpha$, though their estimation error variance also increases with the distance between group means. A similar pattern is observed for the direct effect $\hat{\bar\zeta}$, where transport-based estimates exhibit increasing bias (under-estimation) and variance with larger values of $\alpha$. Interestingly, the behavior of the total effect $\hat{\bar\tau}$ departs from this pattern. With the exception of the penalized transport approach (SKH), the variance of the estimation error for the total effect remains stable as $\alpha$ increases. This suggests that the opposing biases observed in the estimation of direct and indirect effects tend to cancel out, yielding robust total effect estimates.


\begin{figure}[htb!]
    \centering
    \includegraphics[width=\linewidth]{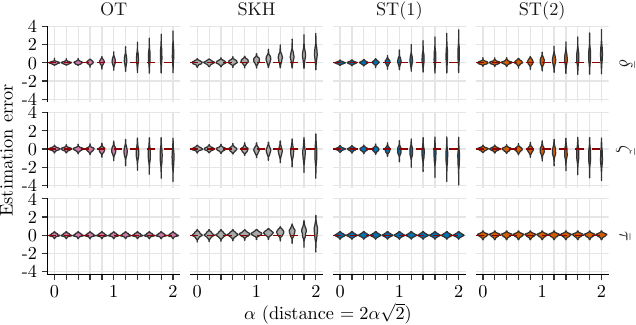}
    \caption{Estimation error of the average causal effects ($\bar{\delta}$,  $\bar{\zeta}$, $\bar{\tau}$) across 200 replicated samples ($n = 500$), in the Gaussian case with two mediators (toy example of section~\ref{app:toy}), for increasing values Of $\alpha$.}
    \label{fig:gaussian-mc-alpha}
\end{figure}

\subsubsection*{Varying the proportion of untreated}

We now turn to the baseline parameter setting ($\alpha=1$) and explore the impact of sample imbalance by varying the proportion of untreated individuals. Specifically, we vary the proportion $n_0/n$ from 5\% to 95\% in increments of 5 percentage points. As in Section~\ref{app:toy-alpha}, we conduct Monte Carlo simulations at each level of imbalance, computing the estimated average indirect, direct, and total causal effects ($\hat{\bar\delta}$, $\hat{\bar\zeta}$, and $\hat{\bar\tau}$, respectively). It is important to note that changing the proportion of untreated individuals does not affect the theoretical values of these quantities, which remain the same as those defined in the balanced (baseline) setting.

The distributions of estimation errors are shown in Figure~\ref{fig:gaussian-mc-prop} for the average indirect effect (top), the average direct effect (middle), and the total effect (bottom).

Overall, across all methods and causal effects, the variances of the causal effect estimators tend to be higher when the proportion of $n_0$ is either very low or very high (around 5\% or 95\%). This stems from the limited reliability of the outcome models $\hat{\mu}_0$ and $\hat{\mu}_1$ when one of the treatment groups is underrepresented. 


In the OT approach, the transport plan relies on theoretical means and covariances within the treated and untreated groups. Therefore, the observed changes in variance and bias of the estimation errors for the three causal effects stem solely from the estimation of the RF models in groups 0 and 1. For $\hat{\bar\delta}$, only $\hat{\mu}_0$ is used, so increasing $n_0$ reduces both variance and bias. For $\hat{\bar\zeta}$, both $\hat{\mu}_0$ and $\hat{\mu}_1$ are required; as $n_0$ increases, $n_1$ decreases, leading to higher variance at extreme proportions (e.g., 5\% or 95\%). The same pattern holds for the total effect $\hat{\bar\tau}$, which also depends on both models. However, the bias of $\hat{\bar\tau}$ remains close to zero across all proportions of $n_0$ because the overestimation of $\hat{\bar\delta}$ and the underestimation of $\hat{\bar\zeta}$ cancel each other out.

The entropic-regularized optimal transport approach (SKH) yield results similar to OT in terms of bias and variance of estimation errors across all three causal effects, with one exception: for $\hat{\bar\tau}$, SKH exhibit the highest variance when the proportion of $n_0$ is 95\%, whereas OT shows the highest variance at 5\%. Similar patterns are observed for the ST approaches (both ST(1) and ST(2)), where the behavior of biases and variances closely matches that of OT across all three causal effects.

\begin{figure}[htb!]
    \centering
    \includegraphics[width=\linewidth]{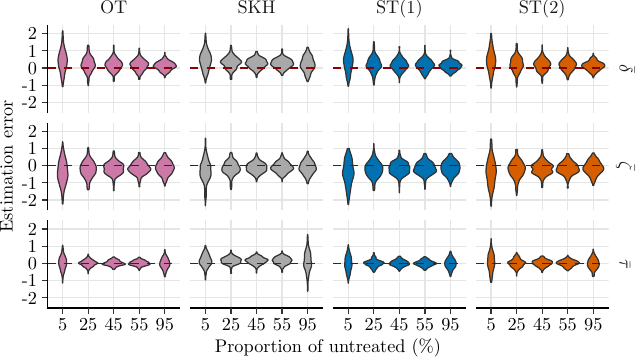}
    \caption{Estimation error of the average causal effects ($\bar{\delta}$, $\bar{\zeta}$, $\bar{\tau}$) across 200 replicated samples ($n = 500$), in the Gaussian case with two mediators (toy example of section~\ref{app:toy}), for increasing proportion of untreated.}
    \label{fig:gaussian-mc-prop}
\end{figure}

\section{Toy example with 3 mediators}\label{app:toy-example-three-mediators}

We simulate a dataset consisting of a binary treatment indicator $A \in \{0, 1\}$, a binary outcome $Y \in \{0, 1\}$, and three covariates: two continuous variables $X_1, X_2 \in \mathbb{R}$, and one categorical variable $X_3 \in \{\text{A}, \text{B}, \text{C}\}$. For individuals in group 0 ($A = 0$), the continuous covariates $(X_1, X_2)$ are drawn from a bivariate normal distribution with mean $\mu_0 = (-1,-1)$ and covariance $\Sigma_0 = 1.2^2 \begin{pmatrix} 1 & 0.5 \\ 0.5 & 1 \end{pmatrix}$. For individuals in group 1 ($A = 1$), we use $\mu_1 = (1.5, 1.5)$ and $\Sigma_1 = 0.9^2 \begin{pmatrix} 1 & -0.4 \\ -0.4 & 1 \end{pmatrix}$, inducing differences in both location and correlation structure across groups.

The categorical covariate $X_3$ is generated conditionally on $(A,X_1,X_2)$, following a multinomial logistic model. Letting $l_{\text{A}}, l_{\text{B}}, l_{\text{C}}$ denote the unnormalized logits for each level of $X_3$, we define:%
$$
\begin{cases}
l_{\text{A}} &= 0.5 + 0.3 X_1 - 0.4 X_2 + 0.2 A, \\
l_{\text{B}} &= -0.3 + 0.5 X_2 - 0.2 X_1 - 0.1 A, \\
l_{\text{C}} &= 0,
\end{cases}
$$
with probabilities obtained via softmax normalization, for $k$ in $\{\text{A}, \text{B}, \text{C}\}$:%
$$
\mathbb{P}(X_3 = k) = \frac{\exp(l_k)}{\exp(l_{\text{A}}) + \exp(l_{\text{B}}) + \exp(l_{\text{C}})}.
$$

The outcome $Y$ is simulated using logistic regression, with functional forms differing by treatment group. We define the logit model:
$$
\begin{cases}
\eta_0 = -0.2 + 0.6 X_1 - 0.6 X_2 + \gamma(X_3), & \text{for } A=0,\\
\eta_1 = 0.1 - 0.2 X_1 + 0.8 X_2 + \gamma(X_3), & \text{for } A=1,
\end{cases}
$$
where the contribution of the categorical covariate is given by:
$$
\gamma(X_3, A) = 
\begin{cases}
0.2 \cdot \mathrm{1}_{\{X_3 = \text{B}\}} - 0.3 \cdot \mathrm{1}_{\{X_3 = \text{C}\}}, & \text{for } A=0,\\
-0.2 \cdot \mathrm{1}_{\{X_3 = \text{B}\}} - 0.1 \cdot \mathrm{1}_{\{X_3 = \text{C}\}}& \text{for } A=1.
\end{cases}
$$

The outcome is drawn from a Bernoulli distribution with success probability $\mathbb{P}[Y = 1] = \text{logit}^{-1}(\eta_A)$.

We generate $n_0 = 400$ observations in group~0 and $n_1 = 200$ observations in group~1.
\end{document}